\documentclass[secnum,leqno]{rims-bessatsu}
\usepackage{amsfonts}
\usepackage{amssymb, amsmath}
\usepackage{graphics}
\usepackage[dvips]{graphicx}
\usepackage{eepic}
\usepackage{epic}
\usepackage[enableskew]{youngtab}
\usepackage[enableskew]{youngtab}

\newtheorem{thm}{Theorem}[section]
\newtheorem{crl}[thm]{Corollary}

\newtheorem{lmm}[thm]{Lemma}

\newtheorem{prp}[thm]{Proposition}

\theoremstyle{definition}
\newtheorem{dfn}[thm]{Definition}
\newtheorem{exa}[thm]{Example}
\theoremstyle{remark}
\newtheorem*{rem}{Remark}
\input{tcilatex}
\TitleHead{The su($n$) WZNW fusion ring as integrable model}
\AuthorHead{Christian Korff}
\support{The author is financially supported by a University Research
Fellowship of the Royal Society}
\VolumeNo{x}
\YearNo{2011}
\PagesNo{000--000}
\communication{Received November 30, 2010.
Revised March 21, 2011.}

\begin{document}

\title{The su($n$) WZNW fusion ring as integrable model: a new algorithm to
compute fusion coefficients}
\author{\textsc{Christian Korff}\thanks{%
School of Mathematics \& Statistics, University of Glasgow, Scotland, UK.%
\newline
e-mail: \texttt{christian.korff@glasgow.ac.uk}}}
\maketitle

\begin{abstract}
%optional
This is a proceedings article reviewing a recent combinatorial construction
of the $\widehat{\mathfrak{su}}(n)_k$ WZNW fusion ring by C. Stroppel and
the author. It contains one novel aspect: the explicit derivation of an
algorithm for the computation of fusion coefficients different from the
Kac-Walton formula. The discussion is presented from the point of view of a
vertex model in statistical mechanics whose partition function generates the
fusion coefficients. The statistical model can be shown to be integrable by
linking its transfer matrix to a particular solution of the Yang-Baxter
equation. This transfer matrix can be identified with the generating
function of an (infinite) set of polynomials in a noncommutative alphabet:
the generators of the local affine plactic algebra. The latter is a
generalisation of the plactic algebra occurring in the context of the
Robinson-Schensted correspondence. One can define analogues of Schur
polynomials in this noncommutative alphabet which become identical to the
fusion matrices when represented as endomorphisms over the state space of
the integrable model. Crucial is the construction of an eigenbasis, the
Bethe vectors, which are the idempotents of the fusion algebra.
\end{abstract}

%
% The text goes here.
% Be sure to use the appropriate "theorem-like" environment as
% is the following examples.  Never use plain TeX commands for these, as
% they will cause interference with the styles of other papers.

%\tableofcontents      %optional

\section{Introduction}

Wess-Zumino-Novikov-Witten (WZNW) models are an important class of conformal
field theories (CFT) distinguished by their Lie algebraic symmetry. Due to
this symmetry the primary fields of WZNW theories are in one-to-one
correspondence with the integrable highest weight representations of an
affine Lie algebra; see e.g. the text book \cite{CFTbook} for details and
references. Consider the $\widehat{\mathfrak{su}}(n)_{k}$ WZNW model, then
the set of all \emph{dominant integral weights of level} $k\in \mathbb{Z}%
_{\geq 0}$ is given by
\begin{equation}
P_{k}^{+}=\left\{ \hat{\lambda}=\sum_{i=1}^{n}m_{i}\hat{\omega}%
_{i}\;\left\vert \;\sum_{i=1}^{n}m_{i}=k\right. \!,\;m_{i}\in \mathbb{Z}%
_{\geq 0}\right\}  \label{domweights}
\end{equation}%
where the $\hat{\omega}_{i}$'s denote the fundamental affine weights of the
affine Lie algebra $\widehat{\mathfrak{su}}(n)$; see e.g. \cite{Kac} for
details. Note that we use the label $n$ instead of $0$ for the affine node.
In what follows it will be convenient to identify elements in the set $%
P_{k}^{+}$ with the partitions $\mathcal{P}_{\leq n-1,k}$ whose Young
diagram fits into a bounding box of height $n-1$ and width $k$. Namely,
define a bijection $P_{k}^{+}\rightarrow \mathcal{P}_{\leq n-1,k}$ by
setting
\begin{equation}
\hat{\lambda}\mapsto \lambda =(\lambda _{1},\ldots ,\lambda _{n-1})\quad
\text{with}\quad \lambda _{i}-\lambda _{i+1}=m_{i}~,  \label{weight2part}
\end{equation}%
where $m_{i}$ is the so-called Dynkin label, i.e. the coefficient of the $i^{%
\text{th}}$ fundamental weight in \eqref{domweights}. Vice versa, given a
partition $\lambda \in \mathcal{P}_{\leq n-1,k}$ we shall denote by $\hat{%
\lambda}$ the corresponding affine weight in $P_{k}^{+}$.

Since the set of dominant integral weights at fixed level $k$ has
cardinality $|P_{k}^{+}|=\binom{n+k-1}{k}$, WZNW models are so-called \emph{%
rational} conformal field theories, i.e. they have a finite set of primary
fields from which all other fields can be generated. An important ingredient
in the description of rational conformal field theories is the concept of
fusion: in physics terminology one considers the operator product expansion
of two primary fields. While this can be made mathematically precise in the
context of vertex operator algebras and the fusion process can be identified
with the product in the Grothendieck ring of an abelian braided monoidal
category in the context of tilting modules of quantum groups, we will not
use this mathematical framework here.

Consider the free abelian group (with respect to addition) generated by $%
P_{k}^{+}$ and introduce the so-called fusion product%
\begin{equation}
\hat{\lambda}\ast \hat{\mu}=\sum_{\hat{\nu}\in P_{k}^{+}}\mathcal{N}_{\hat{%
\lambda}\hat{\mu}}^{(k)\hat{\nu}}\hat{\nu},  \label{fusiondef}
\end{equation}%
by defining the structure constants $\mathcal{N}_{\hat{\lambda}\hat{\mu}%
}^{(k)\hat{\nu}}\in \mathbb{Z}_{\geq 0}$, called \emph{fusion coefficients},
via the celebrated Verlinde formula \cite{Verlinde}%
\begin{equation}
\mathcal{N}_{\hat{\lambda}\hat{\mu}}^{(k)\hat{\nu}}=\sum_{\hat{\sigma}\in
P_{k}^{+}}\frac{S_{\hat{\lambda}\hat{\sigma}}S_{\hat{\mu}\hat{\sigma}}S_{%
\hat{\nu}\hat{\sigma}}^{-1}}{S_{\hat{\emptyset}\hat{\sigma}}}\;.
\label{Verlinde}
\end{equation}%
Here $\hat{\emptyset}$ denotes the weight corresponding to the empty
partition and $S$ is the modular $S$-matrix describing the modular
transformation $\tau \rightarrow -\tau ^{-1}$ of affine characters. Among
other properties it enjoys unitarity $\bar{S}_{\hat{\nu}\hat{\sigma}}=S_{%
\hat{\nu}\hat{\sigma}}^{-1}$ and crossing symmetry $\bar{S}_{\hat{\lambda}%
\hat{\sigma}}=S_{\hat{\lambda}^{\ast }\hat{\sigma}},$ where $\hat{\lambda}%
^{\ast }$ is the dual weight obtained by taking the complement of $\lambda $
in the $(n-1)\times k$ bounding box and then deleting all $n$-columns, $%
\lambda ^{\ast }=(\lambda _{1},\lambda _{1}-\lambda _{n-1},\lambda
_{1}-\lambda _{n-2},\ldots ,\lambda _{1}-\lambda _{2})$.

For WZNW models the explicit expression for $S$ is known: the Kac-Peterson
formula \cite{KacPeterson} states $S$ in terms of the Weyl group $W$ and
specialises for $\widehat{\mathfrak{su}}(n)$ to the expression%
\begin{equation}
S_{\hat{\lambda}\hat{\sigma}}=\frac{e^{i\pi n(n-1)/4}}{\sqrt{n(k+n)^{n-1}}}%
\sum_{w\in W}(-1)^{\ell (w)}e^{-\frac{2\pi i}{k+n}(\sigma +\rho ,w(\lambda
+\rho ))}\;,  \label{modS}
\end{equation}%
where $\rho $ is the Weyl vector and $\lambda ,\sigma $ denote the finite,
non-affine weights corresponding to $\hat{\lambda},\hat{\sigma}$. From this
formula it is by no means obvious that the fusion coefficients (\ref%
{Verlinde}) are non-negative integers, however they have been identified
with certain dimensions or multiplicities in various different contexts as
e.g. discussed in \cite{GoodmanNakanishi} (for references see \emph{loc. cit.%
}): dimensions of spaces of conformal blocks of 3-point functions, so-called
moduli spaces of generalised $\theta $-functions; outer multiplicities of
truncated tensor products of tilting modules of quantum groups at roots of
unity; Littlewood-Richardson coefficients of Hecke algebras at roots of
unity; dimensions of local states in restricted-solid-on-solid models. In
fact, (\ref{fusiondef}) defines a unital, commutative ring over the integers
$\mathbb{Z}$, which we shall refer to as the $\widehat{\mathfrak{su}}(n)$
\emph{fusion ring at level k}, denoted by $\mathcal{F}_{n,k}$, and to the
corresponding unital, commutative and associative algebra $\mathcal{F}%
_{n,k}^{\mathbb{C}}\mathbb{=\mathcal{F}}_{n,k}\mathbb{\otimes }_{\mathbb{Z}}%
\mathbb{C}$ as \emph{fusion} or \emph{Verlinde} \emph{algebra}.

This article aims to give a non-technical account of the main findings in
\cite{KS}\ and \cite{Korff} regarding the $\widehat{\mathfrak{su}}(n)$
fusion ring. For proofs the reader is referred to the mentioned papers.
Sections 2 and 3 are largely a summary of previous results reviewing the
definition of an integrable statistical mechanics model which generates the
fusion ring. It is convenient to describe the statistical model and its
lattice configurations using non-intersecting paths, since this allows for
instance a non-technical definition of the transfer matrix. However, it
needs to be stressed that at the moment the path picture is not used to give
combinatorial proofs, instead the discussion is algebraic and employs the
solution to the Yang-Baxter equation given in \cite{Korff}. However, we
present one result, Corollary \ref{path_count}, which relates the counting
of non-intersecting paths on the cylinder to a sum over fusion coefficients.
We also make contact with the phase model of Bogoliubov, Izergin and
Kitanine \cite{BIK} where closely related algebraic structures have been
discussed. Section 4 states a detailed derivation of the new algorithm to
compute fusion coefficients. First a review of the Bethe ansatz equations is
given by highlighting how they are connected to a fusion potential. The
latter differs from the familiar fusion potential of Gepner \cite{Gepner}
and the algorithm therefore yields expressions for fusion coefficients in
terms of Littlewood-Richardson coefficients which differ from the ones
obtained via the celebrated Kac-Walton formula \cite{Kac} \cite{Walton};
compare also with the work of Goodman and Wenzl \cite{GoodmanWenzl}. Section
5 stresses that the Bethe vectors constructed via the quantum inverse
scattering method can be identified with the idempotents of the fusion ring.
The proof of this result has been given before \cite{KS} but their role has
not been emphasized. The modular S-matrix is the transition matrix from the
basis of integrable weights to the basis of Bethe vectors and, hence, can be
expressed in terms of the generators of a Yang-Baxter algebra. We also
present a new perspective on the affine plactic Schur polynomials which are
defined in terms of the transfer matrix via a determinant formula: they
constitute the set of conserved quantities of the integrable model and hence
should be seen as the quantum analogue of a spectral curve. In the present
model this quantum spectral curve coincides with the collection of $\widehat{%
\mathfrak{su}}(n)_{k}$ fusion rings for $k\in \mathbb{Z}_{\geq 0}$. We
conclude with some practical applications, recursion formulae for fusion
coefficients at different level. The last section discusses how the findings
summarised here might generalise to a wider class of integrable models.

\section{Fusion coefficients from statistical mechanics}

We start by defining a statistical vertex model which is obtained in the
crystal limit of $U_{q}\widehat{\mathfrak{su}}(2)$; see \cite{Korff}.
Consider a $n\times (n-1)$ square lattice with quasi-periodic boundary
conditions in the horizontal direction, i.e. a square lattice on a cylinder
with $n-1$ rows. On the edges of the square lattice live statistical
variables $m\in \mathbb{Z}_{\mathbb{\geq }0}$, which we will identify with
the Dynkin labels of dominant integrable weights below. To each lattice
configuration we assign a \textquotedblleft Boltzmann
weight\textquotedblright\ in $\mathbb{C}[z][x_{1},\ldots ,x_{n-1}]$ by
taking the product over (local) vertex configurations.

Label the statistical variables $a,b,c,d\in \mathbb{Z}_{\mathbb{\geq }0}$
sitting on the edges of a vertex in the $i^{\text{th}}$ lattice row as shown
in Figure \ref{fig:boltzmann}, then we assign to it the weight (compare with
\cite{Korff})
\begin{equation}
\mathcal{R}_{c,d}^{a,b}(x_{i})=\left\{
\begin{array}{cc}
x_{i}^{a}, & d=a+b-c,\;b\geq c \\
0, & \text{else}%
\end{array}%
\right. \;.  \label{Boltzmannweight}
\end{equation}%
It is convenient to describe the vertex configurations in terms of
non-intersecting paths, so-called $\infty $-friendly walkers; compare with
\cite{GesselKrattenthaler,Guttmann2,Guttmann3}. In Figure \ref{fig:boltzmann}
$b$ walkers are entering the vertex from above turning to the right, at
which point a contingent of them of size $b-c$ chooses to defect. The
defectors then join another group of $a$ walkers coming from the left. The
Boltzmann weight of the $i^{\text{th}}$ lattice row is then given by $%
z^{a_{n}}\prod_{1\leq j\leq n}\mathcal{R}_{c_{j},d_{j}}^{a_{j},b_{j}}(x_{i})$%
, where we have introduced a parameter $z$ which keeps track of how many
walkers pass the boundary. Figure \ref{fig:Qrowconfig} shows on a simple
example that the weight of a single row is easily computed by simply
counting the number of horizontal edges.

\begin{figure}[tbp]
\begin{equation*}
\includegraphics[scale=0.35]{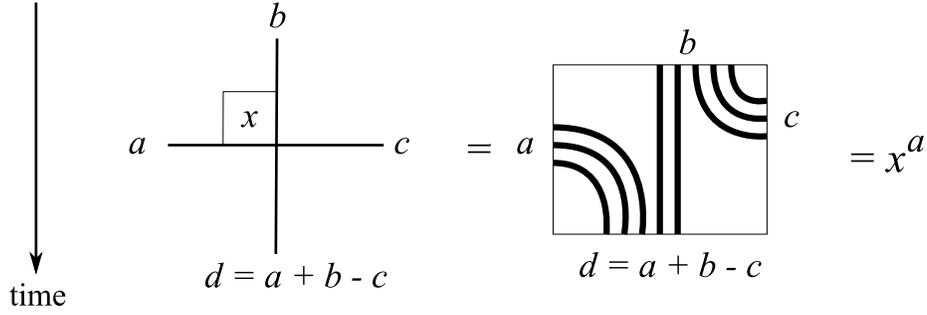}
\end{equation*}%
\caption{Graphical depiction of a vertex configuration with Boltzmann weight
(\protect\ref{Boltzmannweight}). The statistical variables $a,b,c,d\in
\mathbb{Z}_{\geq 0}$ obey the constraints $a+b=c+d$ and $b\geq c$.}
\label{fig:boltzmann}
\end{figure}
Given two arbitrary but fixed affine weights $\hat{\mu},\hat{\nu}\in
P_{k}^{+}$, denote by $m(\hat{\mu}),m(\hat{\nu})$ their $n$-tuples of Dynkin
labels in (\ref{domweights}). Denote by $\Gamma _{\hat{\nu}}^{\hat{\mu}}$
the lattice configurations\ where the outer vertical edges at the bottom and
top of the cylinder take the values $m(\hat{\mu}),m(\hat{\nu})$,
respectively. Figure \ref{fig:lattice_config} shows on an example that each
lattice configuration corresponds to $k$ nonintersecting paths some of which
are closely bunched together. The corresponding partition function of the
vertex model, i.e. the weighted sum over all lattice configurations is
defined as%
\begin{equation}
Z_{\hat{\nu}}^{\hat{\mu}}(x_{1},\ldots
,x_{n-1};z):=\sum_{\{(a_{ij},b_{ij},c_{ij},d_{ij})\}\in \Gamma _{\hat{\mu}}^{%
\hat{\nu}}}\prod_{1\leq i\leq n-1}z^{a_{in}}\prod_{1\leq j\leq n}\mathcal{R}%
_{c_{ij},d_{ij}}^{a_{ij},b_{ij}}(x_{i})\;  \label{Z}
\end{equation}%
with $(a_{ij},b_{ij},c_{ij},d_{ij})$ denoting the vertex configuration in
the $i^{th}$ row and $j^{th}$ column. As we will see below the partition
function is symmetric in the variables $x_{i}$ and, therefore, can be
expanded into a suitable basis in the ring of symmetric functions $\mathbb{C}%
[z][x_{1},\ldots ,x_{n-1}]^{\mathbb{S}_{n-1}}$. We choose the basis of Schur
functions $\{s_{\lambda }\}$ where $\lambda $ is a partition with length $%
\ell (\lambda )<n$.

We remind the reader that the Schur function $s_{\lambda
}(x_{1},...,x_{n-1}) $ can be defined as weighted sum of Young tableaux.
Given a partition $\lambda $, a Young tableau $t$ of shape $|t|=\lambda $ is
a filling of the Young diagram with integers in the set $\{1,\ldots ,n-1\}$
such that the numbers are weakly increasing in each row from left to right
and are strictly increasing in each column from top to bottom. To each
tableau we assign the weight vector $\alpha =(\alpha _{1},\ldots ,\alpha
_{n-1})$ where $\alpha _{i}$ is the multiplicity of $i$ occuring in $t$. The
Schur function is then given as $s_{\lambda
}(x_{1},...,x_{n-1})=\sum_{|t|=\lambda }x_{1}^{\alpha _{1}(t)}\cdots
x_{n-1}^{\alpha _{n-1}(t)}$. We state an explicit example.

\begin{exa}
\textrm{Let $n=3$ and $\lambda =(2,1)$. Then the list of possible tableaux $%
t $ reads,%
\begin{equation*}
t=\Yvcentermath1{\young(11,2),\;\young(11,3),\;\young(12,2),\;\young(13,2),\;%
\young(12,3),\;\young(13,3),\;\young(22,3),\;\young(23,3)\;.}
\end{equation*}%
Thus, the Schur function is the following polynomial%
\begin{equation*}
s_{(2,1)}=x_{1}^{2}x_{2}+x_{1}^{2}x_{3}+x_{1}x_{2}^{2}+2x_{1}x_{2}x_{3}+x_{1}x_{3}^{2}+x_{2}^{2}x_{3}+x_{2}x_{3}^{2}\;.
\end{equation*}
}
\end{exa}

Expanding the partition function (\ref{Z}) with respect to Schur functions
we obtain a relation between the statistical mechanics model defined via (%
\ref{Boltzmannweight}) and the fusion algebra of the $\widehat{\mathfrak{su}}%
(n)_{k}$-WZNW model \cite{Korff}.

\begin{prp}[generating function for fusion coefficients]
\label{Zfusion} The partition function (\ref{Z}) has the expansion%
\begin{equation}
Z_{\hat{\nu}}^{\hat{\mu}}(x_{1},\ldots ,x_{n-1};z)=\sum_{\hat{\lambda}\in
P_{k}^{+}}z^{d}\mathcal{N}_{\hat{\lambda}\hat{\mu}}^{(k),\hat{\nu}%
}s_{\lambda }(x_{1},...,x_{n-1})\;,  \label{Zvertex}
\end{equation}%
where $\mathcal{N}_{\hat{\lambda}\hat{\mu}}^{(k),\hat{\nu}}$ are the fusion
coefficients and the degree $d$ is given by $d:=\frac{|\lambda |+|\mu |-|\nu
|}{n}+\nu _{1}-\mu _{1}$.
\end{prp}

%\begin{rem}\rm
%The partition function (\ref{Z}) should be understood as a generalised skew
%Schur function. Namely, if we were to set $z=0$ only paths contribute which
%do not transgress the boundary. Then only those terms in (\ref{Zvertex})
%survive for which $|\lambda |+|\hat{\mu}|-|\hat{\nu}|=0$ and, as we will
%discuss in more detail below, the fusion coefficients become
%Littlewood-Richardson coefficients, $\mathcal{N}_{\hat{\lambda}\hat{\mu}%
%}^{(k),\hat{\nu}}=\delta _{\mu _{1},\nu _{1}}c_{\hat{\mu},\lambda }^{\hat{\nu%
%}}$. Hence, under this specialisation the expansion (\ref{Zvertex}) is the
%familiar expansion of the skew Schur polynomial $s_{\hat{\nu}/\hat{\mu}}$
%into ordinary Schur polynomials; see \cite{MacDonald}. For $z>0$ the
%resulting polynomials are related to Postnikov's cylindric Schur functions
%\cite{Postnikov} ; details will be presented elsewhere.
%\end{rem}

Given a square $s$ at position $(i,j)$ in the Young diagram of a partition $%
\lambda $, recall that the hook length is defined as $h(s)=\lambda
_{i}+\lambda _{j}^{t}-i-j+1$. That is, $h(s)$ is the number of squares to
right in the same row and the number of squares in the column below it plus
one (for the square itself). The content of the same square is simply
defined as $c(s)=j-i$. Denote by $\Gamma _{\hat{\nu}}^{\hat{\mu}}(d)\subset
\Gamma _{\hat{\nu}}^{\hat{\mu}}$ the subset of lattice path configurations
which have a fixed number of $2d$ outer horizontal edges. The following
result on the number of possible lattice configurations on the cylinder
appears to be new.

\begin{crl}[lattice configurations and fusion coefficients]
\label{path_count} Specialising to $z=x_{1}=\cdots =x_{n-1}=1$ in (\ref%
{Zvertex}) we obtain the identity%
\begin{equation}
|\Gamma _{\hat{\nu}}^{\hat{\mu}}(d)|=\sum_{\hat{\lambda}\in P_{k}^{+}}%
\mathcal{N}_{\hat{\lambda}\hat{\mu}}^{(k),\hat{\nu}}\prod_{s\in \lambda }%
\frac{n-1+c(s)}{h(s)}\;,
\end{equation}%
where the sum can be restricted to those weights $\hat{\lambda}$ for which $%
d=\frac{|\lambda |+|\mu |-|\nu |}{n}+\nu _{1}-\mu _{1}$.
\end{crl}

\begin{proof}
The assertion follows from the well-known formula \cite[Chapter I, Section 3, Example 4]{MacDonald}, $s_\lambda(1,\ldots,1)=\prod_{s\in\lambda}(n-1+c(s))/h(s)$ and the previous expansion (\ref{Zvertex}) of the partition function.
\end{proof}

\begin{rem}
\textrm{While in \cite{GesselKrattenthaler,Guttmann2,Guttmann3} the
Gessel-Viennot method has been used to obtain analogous results for
different type of boundary conditions, the proof of Proposition \ref{Zfusion}
rests on the Yang-Baxter equation, see (\ref{RQQ}) below, and the quantum
inverse scattering method, thus it is of algebraic nature.}
\end{rem}

%\subsection{Bijection between paths and toric skew tableaux}

\begin{figure}[tbp]
\begin{equation*}
\includegraphics[scale=0.65]{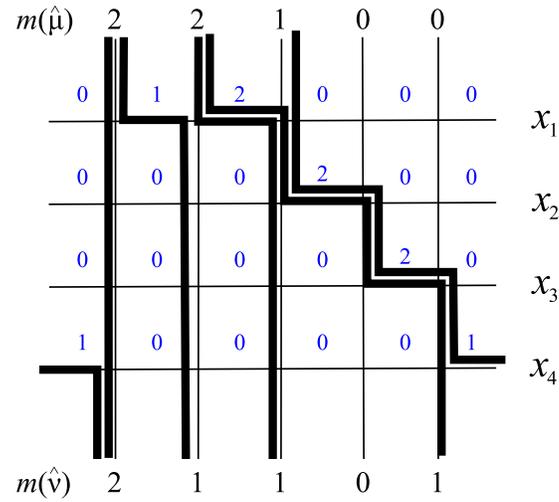}
\end{equation*}%
\caption{An example of a lattice configuration in terms of non-intersecting
paths for $n=k=5$ and $\protect\mu=(5,3,1)$, $\protect\nu =(4,2,1)$.}
\label{fig:lattice_config}
\end{figure}

\section{Transfer matrix and Yang-Baxter algebras}

\begin{figure}[tbp]
\begin{equation*}
\includegraphics[scale=0.75]{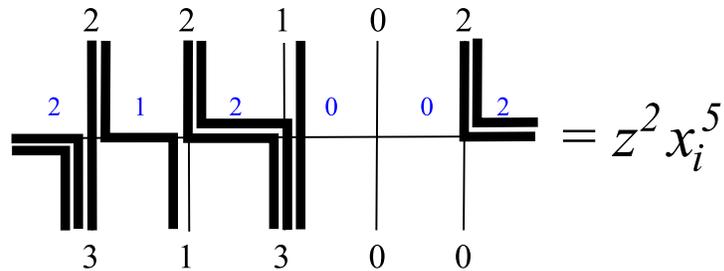}
\end{equation*}%
\caption{An example for $n=5$ and $k=7$ of a row configuration and its
"statistical weight", which is obtained by counting the horizontal edges of
paths. Outer edges are identified and for each a power of $z$ is added.}
\label{fig:Qrowconfig}
\end{figure}

We now introduce the row-to-row transfer matrix of the vertex model (\ref%
{Boltzmannweight}) as the partition function of a single lattice row and
then identify it below as generating function of certain polynomials in a
noncommutative alphabet.

\begin{dfn}[transfer matrix]
Given any two $n$-tuples $\boldsymbol{m}=(m_{1},\ldots ,m_{n})$ and $%
\boldsymbol{m}^{\prime }=(m_{1}^{\prime },\ldots ,m_{n}^{\prime })\in
\mathbb{Z}_{\geq 0}^{n}$, the transfer matrix $Q(x_{i})$ of the $i^{\text{th}%
}$ row is defined via the elements
\begin{equation}
Q(x_{i})_{m,m^{\prime }}:=\sum_{\substack{ allowed\,row\,  \\ configurations
}}z^{\frac{\text{\# of outer edges}}{2}}x_{i}^{\text{\# of horizontal edges}%
},  \label{Qdef}
\end{equation}%
where the factor 1/2 in the power of the variable $z$ takes into account
that the outer horizontal edges need to be identified, since we are on the
cylinder.
\end{dfn}

As it is common in the discussion of vertex models we wish to identify the
transfer matrix as an endomorphism of a vector space. For this purpose we
now interpret the statistical variables at the lattice edges as labels of
basis vectors in the vector space $\mathcal{M}=\tbigoplus_{m\in \mathbb{Z}_{%
\mathbb{\geq }0}}\mathbb{C}v_{m}$. Then a row configuration in the lattice,
i.e. an assignment of statistical variables $\boldsymbol{m}=(m_{1},\ldots
,m_{n})$ along one row of vertical edges, fixes a vector $v_{m_{1}}\otimes
\cdots \otimes v_{m_{n}}\in \mathcal{M}^{\otimes n}$. Henceforth, we
identify the tensor product $\mathcal{M}^{\otimes n}$ with $\mathbb{C}P^{+}$%
, the complex linear span of \emph{all} the integral dominant weights of the
affine Lie algebra $\widehat{\mathfrak{su}}(n)$, $P^{+}:=\left\{ \left. \hat{%
\lambda}=\tsum_{i=1}^{n}m_{i}\hat{\omega}_{i}\right\vert ~m_{i}\in \mathbb{Z}%
_{\geq 0}\right\} ,$ via the map $\boldsymbol{m}\mapsto \tsum_{i=1}^{n}m_{i}%
\hat{\omega}_{i}$. That is, we interpret the statistical variables $%
\boldsymbol{m}=(m_{1},\ldots ,m_{n})$ in one row of our vertex model as
Dynkin labels of an affine weight in $P^{+}$. For convenience we will
sometimes denote this $n$-tuple $\boldsymbol{m}$ and the associated vector
in $\mathcal{M}^{\otimes n}$ by the same symbol. By construction the
row-to-row transfer matrix and the partition function are then related via
\begin{equation}
Z_{\hat{\mu}}^{\hat{\nu}}(x_{1},\ldots ,x_{n-1};z)=\langle \boldsymbol{m}(%
\hat{\nu}),Q(x_{n-1})\cdots Q(x_{1})\boldsymbol{m}(\hat{\mu})\rangle \;,
\label{Z=QQQQ}
\end{equation}%
where we have introduced the inner product $\langle \boldsymbol{m,m}^{\prime
}\rangle :=\prod_{i=1}^{n}\delta _{m_{i},m_{i}^{\prime }}$ which we assume
to be antilinear in the first factor. Thus, the transfer matrix (\ref{Qdef})
can be interpreted as discrete time evolution operator which successively
generates the paths on the cylindric square lattice. Note that for any pair
of configurations $\boldsymbol{m},\boldsymbol{m}^{\prime }\in \mathcal{M}%
^{\otimes n}$ only a finite number of the terms making up the matrix element
$\langle \boldsymbol{m},Q(x_{i})\boldsymbol{m}^{\prime }\rangle $ is
non-zero. The operator $Q\in \limfunc{End}\mathcal{M}^{\otimes n}$ is
therefore well-defined. We now reformulate the transfer matrix in terms of a
set of more elementary, local operators which respectively increase and
decrease a single Dynkin label $m_{i}$ only.

\subsection{The local affine plactic algebra}

For $i=1,\ldots ,n$ define maps $\varphi _{i},\varphi _{i}^{\ast }\in
\limfunc{End}(\mathcal{M}^{\otimes n})$ by setting%
\begin{equation}
\varphi _{i}^{\ast }\boldsymbol{m}=(m_{1},\ldots ,m_{i}+1,\ldots ,m_{n})
\label{phi*}
\end{equation}%
and%
\begin{equation}
\varphi _{i}\boldsymbol{m}=\left\{
\begin{array}{cc}
(m_{1},\ldots ,m_{i}-1,\ldots ,m_{n}), & m_{i}>0 \\
0, & \text{else}%
\end{array}%
\right. \;.  \label{phi}
\end{equation}%
In addition let $N_{i}\boldsymbol{m}=m_{i}~\boldsymbol{m}$ for all $1\leq
i\leq n$. These maps can be identified with the Chevalley generators of the $%
U_{q}\mathfrak{sl}(2)$ Verma module in the crystal limit; see \cite{Korff}.
They have first appeared in the context of the phase model; see \cite{BIK}
and references therein. In \cite{KS} the following statement has been proven
by constructing an explicit basis for the phase algebra $\Phi $.

\begin{prp}[phase algebra]
The $\varphi _{i},\varphi _{i}^{\ast }$ and $N_{i}$ generate a subalgebra $%
\hat{\Phi}$ of $\func{End}(\mathcal{M}^{\otimes n})$ which can be realized
as the algebra $\Phi $ with the following generators and relations for $%
1\leq i,j\leq n$:
\begin{gather}
\varphi _{i}\varphi _{j}=\varphi _{j}\varphi _{i},\quad \varphi _{i}^{\ast
}\varphi _{j}^{\ast }=\varphi _{j}^{\ast }\varphi _{i}^{\ast },\quad
N_{i}N_{j}=N_{j}N_{i}  \label{comm} \\
N_{i}\varphi _{j}-\varphi _{j}N_{i}=-\delta _{ij}\varphi _{i},\quad
N_{i}\varphi _{j}^{\ast }-\varphi _{j}^{\ast }N_{i}=\delta _{ij}\varphi
_{i}^{\ast },  \label{comm2} \\
\varphi _{i}\varphi _{i}^{\ast }=1,\quad \varphi _{i}\varphi _{j}^{\ast
}=\varphi _{j}^{\ast }\varphi _{i}\;\text{ if }\;i\neq j,  \label{comm3} \\
N_{i}(1-\varphi _{i}^{\ast }\varphi _{i})=(1-\varphi _{i}^{\ast }\varphi
_{i})N_{i}=0~.  \label{Npi}
\end{gather}
\end{prp}

Note that with respect to the scalar product introduced above we have $%
\langle \varphi _{i}^{\ast }\boldsymbol{m},\boldsymbol{m}^{\prime }\rangle
=\langle \boldsymbol{m},\varphi _{i}\boldsymbol{m}^{\prime }\rangle \;$for
any $\boldsymbol{m},\boldsymbol{m}^{\prime }\in \mathcal{M}^{\otimes n}$.

\begin{dfn}[local\ affine plactic algebra]
Let $\func{Pl}=\func{Pl}(\mathcal{A})$ be the free algebra generated by the
elements of $\mathcal{A}=\{a_{1},a_{2},\ldots a_{n}\}$ modulo the relations
\begin{eqnarray}
a_{i}a_{j}-a_{j}a_{i}=0, &&\text{ if $|i-j|\neq 1\mod n$},  \label{PL1} \\
a_{i+1}a_{i}^{2}=a_{i}a_{i+1}a_{i}, &&a_{i+1}^{2}a_{i}=a_{i+1}a_{i}a_{i+1},
\label{PL2}
\end{eqnarray}%
where \eqref{PL2} are the plactic relations on the circle, i.e. all indices
are defined modulo $n$. Denote by $\func{Pl}_{\func{fin}}=\func{Pl}_{\func{%
fin}}(\mathcal{A}^{\prime })$ the \emph{local finite plactic algebra}
generated from $\mathcal{A}^{\prime }=\{a_{1},a_{2},\ldots a_{n-1}\}$.
\end{dfn}

We recall the following result from \cite[Prop 5.8]{KS}:

\begin{prp}
\label{faithfulness} There is a homomorphism of algebras $\func{Pl}_{\func{%
fin}}\rightarrow \Phi $ such that
\begin{equation}
a_{j}\mapsto \varphi _{j+1}^{\ast }\varphi _{j},\quad j=1,...,n-1,
\label{placticrep}
\end{equation}%
where the representation of the phase algebra $\Phi $ given by (\ref{phi*})
and (\ref{phi}) lifts to a representation of the local plactic algebra $%
\func{Pl}_{\func{fin}}$. Mapping $a_{0}=a_{n}$ to $z\varphi _{1}^{\ast
}\varphi _{n}$ it lifts in addition to a representation of $\func{Pl}$ on $%
\mathcal{M}[z]^{\otimes n}$ with $\mathcal{M}[z]=\mathbb{C}[z]\otimes _{%
\mathbb{C}}\mathcal{M}$ and $z$ an indeterminate. Both representations are
faithful.
\end{prp}

\begin{rem}
\textrm{The finite plactic algebra first appeared in the context of the
Robinson-Schensted correspondence: given a word in a noncommutative alphabet
it can be mapped onto a pair of Young tableaux, usually called $(P,Q)$, by
using the bumping algorithm. $Q$ is the recording tableau encoding the
sequence of bumping processes; see e.g. \cite{FultonYT} for an explanation.
Lascoux and Sch\"{u}tzenberger showed that identifying words which only
differ in their recording tableaux is equivalent to a set of identities of
which (\ref{PL2}) are special cases. The \emph{local} finite plactic algebra
was first considered by Fomin and Greene in \cite{FG}. We recover their case
when specialising to $z=0$. }
\end{rem}

\begin{rem}
\textrm{Note that the action of the affine plactic algebra is blockdiagonal
with respect to the decomposition $\mathbb{C}P^{+}=\tbigoplus_{k\geq 0}%
\mathbb{C}P_{k}^{+}$. In fact each subspace can be represented as a directed
coloured graph where the elements in $P_{k}^{+}$ are the vertices and a
directed edge of colour $i$ between two vertices, $\hat{\mu}\overset{i}{%
\longrightarrow }\hat{\lambda}$, is introduced if $\hat{\lambda}=a_{i}\hat{%
\mu}$. This yields the Kirillov-Reshetikhin crystal graph $\mathcal{B}_{1,k}$
of type A. Setting $z=0$ all edges related to the affine generator $a_{n}$
are removed from the graph and we obtain the $\mathfrak{su}(n)$ crystal
graph of highest weight $k\omega _{1}$, where $\omega _{1}$ is the first
fundamental weight. }
\end{rem}

%\begin{figure}[tbp]
%\begin{equation*}
%\includegraphics[scale=0.3]{rowconfig.eps}
%\end{equation*}%
%\caption{The allowed row configurations of the vertex model (\protect\ref%
%{Boltzmannweight}) with $\protect\varepsilon_i,m_i,m_i-\protect\varepsilon%
%_i\in\mathbb{Z}_{\geq 0}$. Due to the periodic boundary conditions $\protect%
%\varepsilon_n =\protect\varepsilon_0$.}
%\label{fig:rowconfig}
%\end{figure}

\subsection{Yang-Baxter algebras}

For $r,s\in \mathbb{Z}_{\geq 0}$ arbitrary but fixed and $u$ a formal,
invertible variable define $\boldsymbol{Q}(u):\mathcal{M}(u)\otimes \mathcal{%
M}(u)^{\otimes n}\rightarrow \mathcal{M}(u)\otimes \mathcal{M}(u)^{\otimes
n} $ by setting $\boldsymbol{Q}(u)v_{r}\otimes \boldsymbol{m}:=\sum_{s\geq
0}v_{s}\otimes \boldsymbol{Q}_{s,r}(u)\boldsymbol{m}$ where%
\begin{equation}
\boldsymbol{Q}_{s,r}(u):=\sum_{\varepsilon }u^{|\varepsilon |+r}(\varphi
_{1}^{\ast })^{r}a_{1}^{\varepsilon _{1}}\cdots a_{n-1}^{\varepsilon
_{n-1}}\varphi _{n}^{s}  \label{momQ}
\end{equation}%
with the sum running over all compositions $\varepsilon =(\varepsilon
_{1},\ldots ,\varepsilon _{n-1})$. Despite the sums in the definition of $%
\boldsymbol{Q}(u)$ being infinite, only a finite number of terms survive
when acting on a vector in $\mathcal{M}(u)^{\otimes n+1}$, thus the operator
is well-defined. In fact, we have the following \cite{Korff}:

\begin{lmm}
Let $Q$ be the transfer matrix defined in (\ref{Qdef}). Then%
\begin{equation}
Q(x_{i})=\sum_{r\geq 0}z^{r}\boldsymbol{Q}_{r,r}(x_{i})=\sum_{r\geq
0}(zx_{i})^{r}(\varphi _{1}^{\ast })^{r}\boldsymbol{Q}_{0,0}(x_{i})\varphi
_{n}^{r},  \label{Qdef2}
\end{equation}%
in other words for $z=1$ the transfer matrix is the formal trace of the
matrix (\ref{momQ}).
\end{lmm}

Define another operator $\mathcal{R}(u/v):\mathcal{M}(u)\otimes \mathcal{M}%
(v)\rightarrow \mathcal{M}(u)\otimes \mathcal{M}(v)$ via the relation
\begin{equation}
\mathcal{R}(u)~v_{a}\otimes v_{b}=\sum_{c,d\geq 0}\mathcal{R}%
_{c,d}^{a,b}(u)~v_{c}\otimes v_{d},  \label{Smatrixdef}
\end{equation}%
setting%
\begin{equation}
\mathcal{R}_{c,d}^{a,b}(u)=\left\{
\begin{array}{cc}
u^{a}, & c=b,\;d=a \\
u^{a}(1-u), & d=a+b-c,\;b>c \\
0, & \text{else}%
\end{array}%
\right. \;.
\end{equation}

\begin{prp}
The $\mathcal{R}$-matrix (\ref{Smatrixdef}) and the monodromy matrix (\ref%
{momQ}) obey the Yang-Baxter equation,%
\begin{equation}
\mathcal{R}_{12}(u/v)\boldsymbol{Q}_{1}(u)\boldsymbol{Q}_{2}(v)=\boldsymbol{Q%
}_{2}(v)\boldsymbol{Q}_{1}(u)\mathcal{R}_{12}(u/v)\;.  \label{RQQ}
\end{equation}
\end{prp}

\begin{figure}[tbp]
\begin{equation*}
\includegraphics[scale=0.4]{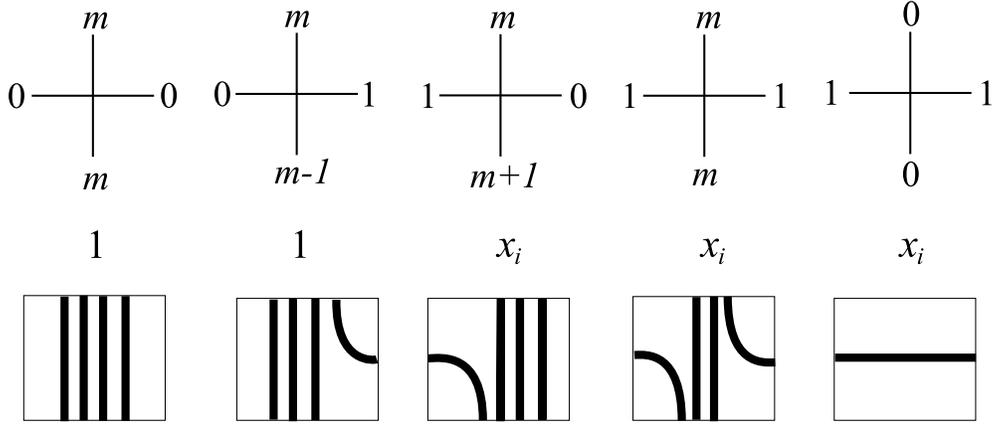}
\end{equation*}%
\caption{The (Boltzmann) weights of local vertex configuration for the
auxiliary matrix $T=A+zD$. Note that in contrast to the Boltzmann weights (%
\protect\ref{Boltzmannweight}) the walkers can now propagate horizontally
but only one is allowed on a horizontal edge at a time.}
\label{fig:Tweights}
\end{figure}
The definition of the monodromy matrix (\ref{momQ}) can be algebraically
motivated by taking a special limit of the intertwiner of a $U_q(\mathfrak{sl%
}(2))$ Verma module; see the discussion in \cite{Korff}. Replacing this
Verma module with the (two-dimensional) fundamental representation one
obtains in the analogous limit now a $2\times 2$ monodromy matrix, $%
\boldsymbol{T}\in \limfunc{End}(\mathbb{C}^{2}\otimes \mathcal{M}%
(u)^{\otimes n})$, where the ordering of the noncommutative alphabet is
reversed($r,s=0,1$),%
\begin{equation}
\boldsymbol{T}_{r,s}(u):=\sum_{\varepsilon _{i}=0,1}u^{|\varepsilon
|+s}\varphi _{n}^{s}a_{n-1}^{\varepsilon _{n-1}}\cdots a_{1}^{\varepsilon
_{1}}(\varphi _{1}^{\ast })^{r}=\left(
\begin{array}{cc}
A(u) & B(u) \\
C(u) & D(u)%
\end{array}%
\right) _{r,s}\;.  \label{momT}
\end{equation}%
Note that the sum now only runs over compositions $\varepsilon $ whose parts
are 0 or 1. This second monodromy matrix coincides with the matrix
introduced by Bogoliubov, Izergin and Kitanine in the context of the
so-called phase model \cite{BIK}.

\begin{prp}[\protect\cite{BIK}]
The $2\times 2$ matrix (\ref{momT}) with entries in $\limfunc{End}(\mathcal{M%
}(u)^{\otimes n})$ obeys the Yang-Baxter equation $R_{12}(u/v)\boldsymbol{T}%
_{1}(u)\boldsymbol{T}_{2}(v)=\boldsymbol{T}_{2}(v)\boldsymbol{T}%
_{1}(u)R_{12}(u/v)$ with%
\begin{equation}
R(u)=\left(
\begin{array}{cccc}
\frac{u}{u-1} & 0 & 0 & 0 \\
0 & 0 & \frac{u}{u-1} & 0 \\
0 & \frac{1}{u-1} & 1 & 0 \\
0 & 0 & 0 & \frac{u}{u-1}%
\end{array}%
\right) \; .  \label{RTT}
\end{equation}
\end{prp}

The matrix elements of (\ref{momT}) obey certain commutation relations with
the matrix elements (\ref{momQ}), which again can be encoded in yet another
solution to the Yang-Baxter equation \cite{Korff}.

\begin{prp}
The monodromy matrices (\ref{momQ}) and (\ref{momT}) satisfy the equation,%
\begin{equation}
L_{12}^{\prime }(u/v)\boldsymbol{T}_{1}(u)\boldsymbol{Q}_{2}(v)=\boldsymbol{Q%
}_{2}(v)\boldsymbol{T}_{1}(u)L_{12}^{\prime }(u/v),  \label{LQT}
\end{equation}%
where the endomorphism $L^{\prime }\in \limfunc{End}(\mathbb{C}^{2}\otimes
\mathcal{M}(u)^{\otimes n})$ is defined in terms the matrix elements%
\begin{equation}
L_{c,d}^{\prime a,b}(u)=\left\{
\begin{array}{cc}
1+u, & a,b,c,d=0 \\
u^{a}, & d=a+b-c,\;b\geq c,\;a,c=0,1 \\
0, & \text{else}%
\end{array}%
\right. \;.  \label{L}
\end{equation}
\end{prp}

\begin{figure}[tbp]
\begin{equation*}
\includegraphics[scale=0.4]{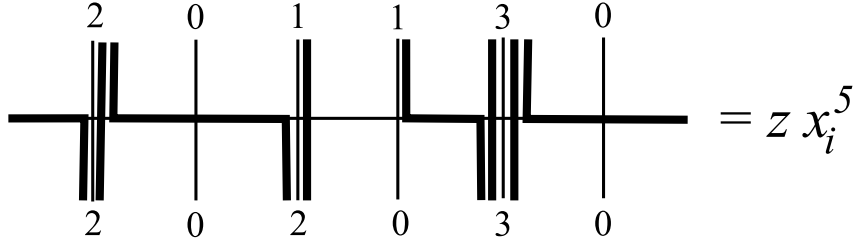}
\end{equation*}%
\caption{An example of a row configuration for the auxiliary matrix $T$ for $%
n=6$ and $k=7$. The statistical weight is again obtained by counting the
horizontal edges of paths, where outer edges are once more identified each
contributing a power of $z$.}
\label{fig:Trowconfig}
\end{figure}
The hallmark of an exactly solvable or integrable model in statistical
mechanics is that its transfer matrix commutes with itself for arbitrary
values of the spectral parameter which here is identified with the variables
$x_{i}$ in each lattice row. In a physical application the row variables $%
\{x_{1},\ldots ,x_{n}\}$ would be evaluated in the interval $[0,1]^{\times
n} $ such that the Boltzmann weights (\ref{Boltzmannweight}) can be
interpreted as proper probabilities. The transfer matrices for any other,
possibly complex values, of the $x_{i}$'s would be seen as a
\textquotedblleft symmetry\textquotedblright\ of the system. Generalising
the notion of Liouville integrability in classical mechanics, such a
statistical model is called integrable. One important consequence of the
Yang-Baxter equations stated above is that they imply integrability of the
vertex model (\ref{Boltzmannweight}).

\begin{crl}[Integrability]
\label{integrability}Set $T=A+zD$ then we have, among others, the
commutation relations%
\begin{equation}
\lbrack Q(u),Q(v)]=[T(u),T(v)]=[T(u),Q(v)]=0,  \label{ncehcomm}
\end{equation}%
where $Q$ is the transfer matrix (\ref{Qdef}). Moreover, the following
relation holds true%
\begin{eqnarray}
\boldsymbol{Q}_{r,r}(u)B(v)&=&\left( 1+\frac{u}{v}\right) \delta _{r,0}B(v)%
\boldsymbol{Q}_{0,0}(u)  \notag \\
&&+\boldsymbol{Q}_{r-1,r}(u)D(v)-\boldsymbol{Q}_{r,r+1}(u)A(v)+\boldsymbol{Q}%
_{r-1,r+1}(u)C(v)  \label{B&Q}
\end{eqnarray}
\end{crl}

Note in particular that the first relation in (\ref{ncehcomm}) entails via (%
\ref{Z=QQQQ}) that the partition function (\ref{Z}) is symmetric in the
variables $x_{i}$ as claimed earlier.

\subsection{Affine plactic elementary and complete symmetric polynomials}

To keep this article self-contained and motivate the definition of the
affine plactic polynomials below, we review some basic facts about symmetric
functions; see \cite{MacDonald} for details. Let $\{y_{1},\ldots ,y_{\ell}\}
$ be a set of commuting variables for some finite $\ell >0$. Recall that the
ring of symmetric functions $\mathbb{C}[y_{1},\ldots ,y_{\ell }]^{\mathbb{S}%
_{\ell }}$ is generated by either the elementary or complete symmetric
functions denoted by the letters $e_{r}$ and $h_{r}$ with $r\in \mathbb{Z}%
_{\geq 0}$, respectively. Both sets of functions can be introduced via the
following generating functions,%
\begin{eqnarray}
\prod_{i=1}^{\ell }(1+y_{i}u) &=&\sum_{r\geq 0}e_{r}(y_{1},\ldots ,y_{\ell
})u^{r}  \label{E} \\
\prod_{i=1}^{\ell }(1-y_{i}u)^{-1} &=&\sum_{r\geq 0}h_{r}(y_{1},\ldots
,y_{\ell })u^{r},  \label{H}
\end{eqnarray}%
respectively. Note that the first sum is finite, i.e. $e_{r}(y_{1},\ldots
,y_{\ell })=0$ for $r>\ell $, while the second one is infinite. Explicitly,
the elementary and complete symmetric functions are given by the expressions%
\begin{eqnarray}
e_{r}(y_{1},\ldots ,y_{\ell }) &=&\sum_{1\leq i_{1}<\cdots <i_{r}\leq \ell
}y_{i_{1}}\cdots y_{i_{r}}\text{\quad }  \label{e} \\
h_{r}(y_{1},\ldots ,y_{\ell }) &=&\sum_{1\leq i_{1}\leq \cdots \leq
i_{r}\leq \ell }y_{i_{r}}\cdots y_{i_{1}}\;.  \label{h}
\end{eqnarray}%
From the generating functions it is immediate to deduce that the $e$'s and $%
h $'s satisfy recursion relations. Namely, one has the identities%
\begin{eqnarray}
e_{r}(y_{1},\ldots ,y_{\ell }) &=&e_{r}(y_{1},\ldots ,y_{\ell -1})+y_{\ell
}e_{r-1}(y_{1},\ldots ,y_{\ell -1}),  \label{erec} \\
h_{r}(y_{1},\ldots ,y_{\ell }) &=&h_{r}(y_{1},\ldots ,y_{\ell -1})+y_{\ell
}h_{r-1}(y_{1},\ldots ,y_{\ell })\;.  \label{hrec}
\end{eqnarray}%
As mentioned above one has $\mathbb{C}[y_{1},\ldots ,y_{\ell }]^{\mathbb{S}%
_{\ell }}\cong \mathbb{C}[e_{1},\ldots ,e_{\ell }]\cong \mathbb{C}%
[h_{1},\ldots ,h_{\ell }]$ and, thus, both sets of functions must be
polynomials of each other. In fact, multiplying both generating functions
(with $u$ replaced by $-u$ in the first one) yields the equations $%
\sum_{r=a+b}(-1)^{a}e_{a}h_{b}=0$ for all $r\geq 0$ which can be solved
either for the $e$'s or the $h$'s to yield the well-known Jacobi-Trudi
formulae%
\begin{eqnarray}
h_{r}(y_{1},\ldots ,y_{\ell }) &=&\det (e_{1-i+j}(y_{1},\ldots ,y_{\ell
}))_{1\leq i,j\leq r},  \label{JTh2e} \\
e_{r}(y_{1},\ldots ,y_{\ell }) &=&\det (h_{1-i+j}(y_{1},\ldots ,y_{\ell
}))_{1\leq i,j\leq r}\;.  \label{JTe2h}
\end{eqnarray}

We will now generalise these functions by replacing the commuting variables $%
\{y_{1},\ldots ,$ $y_{\ell}\}$ with the noncommutative variables $%
\{a_{1},\ldots ,a_{n}\}$, i.e. the generators of the affine plactic algebra.
We wish to identify the auxiliary matrix $T$ defined in Corollary \ref%
{integrability}, the trace of (\ref{momT}), and the transfer matrix $Q$
defined either through (\ref{Qdef}) or (\ref{Qdef2}) as noncommutative
analogues of the generating functions (\ref{E}) and (\ref{H}), respectively.
When trying to generalise (\ref{e}) and (\ref{h}) to a noncommutative
alphabet one needs to specify an ordering of the variables. The auxiliary
and transfer matrix prescribe such an ordering of the noncommutative
variables $\{a_{i}\}$ which turns out to be consistent and has the
additional desirable property that the resulting \emph{affine plactic
elementary} and \emph{complete symmetric polynomials} form a commutative
subalgebra due to the integrability condition (\ref{ncehcomm}). Because this
ordering is cyclic, it is easier to split the affine generator $a_{n}$ into
its constituents $z,\varphi _{1}^{\ast }$ and $\varphi _{n}$, compare with
Proposition \ref{faithfulness}, and write the variables in descending or
ascending order as we have already done in (\ref{momT}) and (\ref{momQ}).

\begin{prp}[generating functions]
The transfer matrix (\ref{Qdef}) can be interpreted as the generating
function%
\begin{equation}
Q(x_{i})=\sum_{r\geq 0}z^{r}\boldsymbol{Q}_{r,r}(x_{i})=\sum_{r\geq
0}x_{i}^{r}h_{r}(\mathcal{A})  \label{Qgenh}
\end{equation}%
of the \emph{affine plactic complete symmetric polynomials}%
\begin{equation}
h_{r}(\mathcal{A}):=\sum_{\varepsilon \vdash r}z^{\varepsilon _{0}}(\varphi
_{1}^{\ast })^{\varepsilon _{0}}a_{1}^{\varepsilon _{1}}\cdots
a_{n-1}^{\varepsilon _{n-1}}\varphi _{n}^{\varepsilon _{0}}\;.  \label{nch}
\end{equation}%
The auxiliary matrix defined in Corollary \ref{integrability},
\begin{equation}
T(x_{i})=A(x_{i})+zD(x_{i})=\sum_{r\geq 0}x_{i}^{r}e_{r}(\mathcal{A}),
\label{T}
\end{equation}%
on the other hand, yields the generating function of the \emph{affine
plactic elementary symmetric polynomials},%
\begin{equation}
e_{r}(\mathcal{A}):=\sum_{\substack{ \varepsilon \vdash r  \\ \varepsilon
_{i}=0,1}}z^{\varepsilon _{n}}\varphi _{n}^{\varepsilon
_{n}}a_{n-1}^{\varepsilon _{n-1}}\cdots a_{1}^{\varepsilon _{1}}(\varphi
_{1}^{\ast })^{\varepsilon _{n}}  \label{nce}
\end{equation}
\end{prp}

\subsection{Combinatorial description of the Yang-Baxter algebra}

Note that when setting the quasi-periodicity parameter $z$ to zero, i.e.
enforcing open boundary conditions, one obtains the finite plactic
polynomials of \cite{FG}. For instance, we have that the matrix element $A$
in (\ref{momT}) is the generating function of the \emph{finite} plactic
elementary polynomials,%
\begin{equation}
z=0:\qquad A(u)=(1+a_{n-1}u)\cdots (1+a_{1}u)=\sum_{r\geq 0}u^{r}e_{r}(%
\mathcal{A}^{\prime })\;.  \label{A}
\end{equation}

\begin{lmm}
The action of the operator $e_{r}(\mathcal{A}^{\prime })=e_{r}(a_{1},\ldots
,a_{n-1})$ on a partition $\mu \in \mathcal{P}_{\leq n-1,k}$ produces a sum
over all partitions $\lambda $ in the $n\times k$ bounding box such that the
skew diagram $\lambda /\mu $ is a vertical strip of length $r$ and $\lambda
_{1}=\mu _{1}$,
\begin{equation}
e_{r}(\mathcal{A}^{\prime })\mu =\sum_{\lambda \in \mathcal{P}_{\leq
n,k},\;\lambda /\mu =(1^{r})}\delta _{\lambda _{1},\mu _{1}}\lambda ^{\prime
}\;.  \label{nce0action}
\end{equation}
Here $\lambda ^{\prime }\in \mathcal{P}_{\leq n-1,k}$ is the partition
obtained from $\lambda $ by deleting all columns of height $n$.
\end{lmm}

\begin{proof}
Using the graphical depiction of the Boltzamnn weights in Figure \ref%
{fig:Tweights} it follows that for $z=0$ only row configurations are allowed
which do not have an occupied outer horizontal edge. Hence, $m_{1}(\mu )\geq
m_{1}(\nu )$ which entails that we must have $\mu _{1}=\nu _{1}$ according
to (\ref{weight2part}). Moreover, we can deduce from the Boltzmann weights
that $m_{i}(\nu )-m_{i}(\mu )=0$ or 1 for $i>1$, hence $\lambda /\mu $ must
be a horizontal $r$-strip.
\end{proof}

In order to describe the action of the remaining matrix elements in (\ref%
{momT}) note that in terms of partitions the map $\varphi _{1}^{\ast }$ adds
a column of height one and increases the width of the bounding box, the
level $k$, by one. The map $\varphi _{n}$ simply decreases the width of the
bounding box if $\lambda _{1}<k$, otherwise it sends $\lambda $ to zero.
Thus, the following formulae%
\begin{equation}
B(u)=u~A(u)\varphi _{1}^{\ast },\qquad C(u)=\varphi _{n}A(u),\qquad
D(u)=u\varphi _{n}A(u)\varphi _{1}^{\ast }~,  \label{BCD}
\end{equation}%
which can be easily checked from (\ref{momT}), allow one to perform
computations with the Yang-Baxter algebra purely in terms of Young diagrams
and their bounding boxes.

\begin{exa}
\textrm{Let $n=5$ and choose $\mu =(2,2,1,0)$ with $\boldsymbol{m}%
=(0,1,1,0,1)$, that is $k=3$. Then we have for $r=3$ only a single term,%
\begin{equation*}
\Yvcentermath1e_{3}(\mathcal{A}^{\prime })~\yng(2,2,1)=\Yvcentermath1\yng%
(2,2,2,1,1)\equiv \Yvcentermath1\yng(1,1,1)\;.
\end{equation*}%
In contrast the action of the \emph{affine} plactic polynomial yields the sum%
\begin{equation*}
\Yvcentermath1e_{3}(\mathcal{A})~\yng(2,2,1)=\Yvcentermath1\yng(1,1,1)~+z~%
\yng(3,3,2)~+z~\yng(3,3,1,1)~+z~\yng(3,2,2,1)~+z~\yng(2,1)\;.
\end{equation*}%
By converting each partition in the above sum to the corresponding
compositions one verifies that each of the last four terms on the right hand
side is generated by a monomial in $\mathcal{A}$ which contains the affine
generator $a_{n}$. For instance, for the second term one finds%
\begin{equation*}
\Yvcentermath1z~\underset{\boldsymbol{m}=(0,1,2,0,0)}{\yng(3,3,2)}=%
\Yvcentermath1z\varphi _{5}a_{2}a_{1}\varphi _{1}^{\ast }~\underset{%
\boldsymbol{m}=(0,1,1,0,1)}{\yng(2,2,1)}=\Yvcentermath1a_{2}a_{1}a_{5}~\yng%
(2,2,1)\;,
\end{equation*}%
where we have used (\ref{comm3}) to rewrite the respective term in $e_{3}(%
\mathcal{A})$ as word in the affine plactic generators. This is always
possible for $r<n$. }
\end{exa}

In light of (\ref{A}) and the last example the definition of the auxiliary
matrix (\ref{T}) can be seen as the noncommutative analogue of (\ref{erec}),
since after expanding with respect to the variable $u$ one arrives at the
identity%
\begin{equation}
e_{r}(\mathcal{A})=e_{r}(\mathcal{A}^{\prime })+z\varphi _{n}e_{r-1}(%
\mathcal{A}^{\prime })\varphi _{1}^{\ast }\;.  \label{ncerec}
\end{equation}%
\textbf{Cylic ordering}. As already alluded to in the last example
comparison with the commutative case (\ref{erec}) is made easier by
realising that for $r<n$ the terms in the second summand can always be
rearranged in cyclic order. To expose the general structure more clearly
consider another example provided by the row configuration depicted in
Figure \ref{fig:Trowconfig} for $n=6$ and $k=7$. The latter corresponds via (%
\ref{nce}) to the monomial%
\begin{equation*}
z\varphi _{6}a_{5}a_{4}a_{2}a_{1}\varphi _{1}^{\ast }=za_{2}a_{1}\varphi
_{6}a_{5}a_{4}\varphi _{1}^{\ast }=a_{2}a_{1}a_{6}a_{5}a_{4},
\end{equation*}%
where we have used once more (\ref{comm3}). The general case is now clear:
monomials in the $a_{i}$ which do neither contain $a_{1}$ or $a_{n}$, or
only one of these generators, are always written in descending order from
left to right. If both, $a_{1}$ and $a_{n},$ occur in the same monomial
write the maximal string of form $a_{l}a_{l-1}\cdots a_{2}a_{1}$ to the left
of the remaining letters which should also be in descending order starting
with $a_{n}$ (although the indices might now \textquotedblleft
jump\textquotedblright by more than one). It follows from the definition (%
\ref{nce}) that for $r=n$ we have $e_{n}(\mathcal{A})=z\cdot 1$.\medskip

We now specialise to the finite plactic complete symmetric polynomials,%
\begin{equation}
z=0:\qquad Q(u)=\boldsymbol{Q}_{0,0}(u)=\sum_{r\geq 0}u^{r}h_{r}(\mathcal{A}%
^{\prime })\; .  \label{Q00}
\end{equation}%
From the definition (\ref{nch}) one easily computes the noncommutative
analogue of the recursion relation (\ref{hrec}),%
\begin{equation}
h_{r}(\mathcal{A})=h_{r}(\mathcal{A}^{\prime })+z\varphi _{1}^{\ast }h_{r-1}(%
\mathcal{A})\varphi _{n}\;.  \label{nchrec}
\end{equation}

\begin{lmm}
Let $r\leq k$ and $\mu \in \mathcal{P}_{\leq n-1,k}$ then%
\begin{equation}
h_{r}(\mathcal{A}^{\prime })\mu =\sum_{\lambda \in \mathcal{P}_{\leq
n,k},\;\lambda /\mu =(r)}\delta _{\lambda _{1},\mu _{1}}\lambda ^{\prime }\;.
\label{nch0action}
\end{equation}%
For $r>k$ we have $h_{r}(\mathcal{A}^{\prime })|_{\mathbb{C}P_{k}^{+}}=0$.
\end{lmm}

\begin{proof}
The assertion follows from a similar line of argument as before, but this
time one uses the Boltzmann weights depicted in Figure \ref{fig:boltzmann}.
Since $z=0$ row configurations with outer edges are prohibited, whence $%
m_{1}(\mu )\geq m_{1}(\nu )$. In contrast to the previous case (\ref%
{nce0action}) a friendly walker now cannot propagate horizontally, however
several are allowed at the same time on the horizontal edges. Thus, we
obtain a horizontal instead of a vertical strip. The last identity is also
clear from the graphical depiction of allowed row configurations: the number
of occupied horizontal edges cannot exceed the number of incoming walkers.
\end{proof}

Note that the affine plactic complete symmetric polynomials can only be
rewritten in (reverse) cyclic order for $r<n$ using the same commutation
relations of the phase algebra as before. For $r>n$ the cyclic ordering
ceases to be well-defined and one has to resort to (\ref{nch}).

Finally, we generalise the last identities from the commutative case, the
Jacobi-Trudi formulae (\ref{JTh2e}) and (\ref{JTe2h}), which are subject of
the next proposition \cite{Korff}.

\begin{prp}[operator functional equation]
The generating functions (\ref{Qgenh}) and (\ref{T}) satisfy the operator
functional relation,%
\begin{equation}
T(-u)Q(u)=1+z(-1)^{n}\sum_{k\geq 0}u^{k+n}h_{k}(\mathcal{A})\pi _{k}\;,
\label{TQ}
\end{equation}%
where $\pi _{k}:\mathbb{C}P^{+}\rightarrow \mathbb{C}P_{k}^{+}$ is the
projector onto the subspace spanned by the weights at level $k$. In
particular, for $r<n+k$ the familiar determinant relations from the
commutative case also hold for the noncommutative elementary and complete
symmetric polynomials,%
\begin{equation}
h_{r}(\mathcal{A})=\det (e_{1-i+j}(\mathcal{A}))_{1\leq i,j\leq r},\qquad
e_{r}(\mathcal{A})=\det (h_{1-i+j}(\mathcal{A}))_{1\leq i,j\leq r}\;,
\label{ncJT}
\end{equation}%
where the determinants are well defined due to (\ref{ncehcomm}).
\end{prp}

\begin{figure}[tbp]
\textrm{%
\begin{equation*}
\includegraphics[scale=0.35]{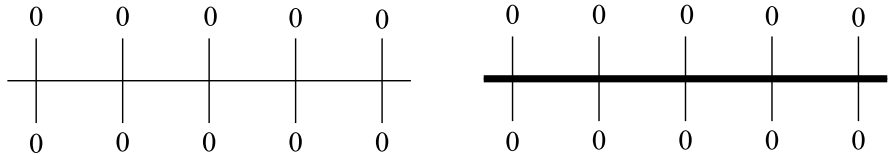}
\end{equation*}%
}
\caption{The two possible vacuum row configurations of the auxiliary matrix $%
T$. Note that the transfer matrix $Q$ only allows for the left one.}
\label{fig:Tvacuum}
\end{figure}
\begin{figure}[tbp]
\textrm{%
\begin{equation*}
\includegraphics[scale=0.6]{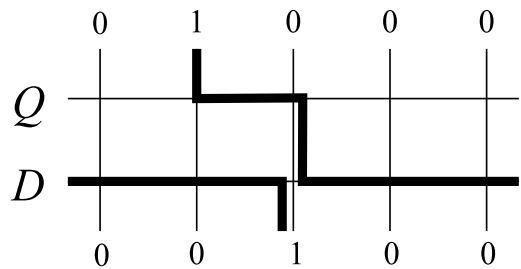}
\end{equation*}%
}
\caption{Depicted is the row configuration resulting from the successive
action of the transfer matrix $Q$ and the operator $D$ on a state of level $%
k=1$ and which is responsible for the additional term on the right hand side
of the functional equation (\protect\ref{TQ}).}
\label{fig:TQk1}
\end{figure}
Setting once more $z=0$ we recover the relation expected from the
commutative case. The additional terms have their origin in the
quasi-periodic boundary conditions and we explain their origin on an example
which will elucidate the general formula.

\begin{exa}
\textrm{Set $n>2$ and consider first the case of level $k=0$. There is only
one state, the \textquotedblleft pseudo-vacuum\textquotedblright\ $\hat{%
\Omega}=(0,0,\ldots ,0)$, and trivially we have $Q(u)\hat{\Omega}=\hat{\Omega%
}$. Acting with $T(-u)$ on $\hat{\Omega}$ we obtain two contributions shown
in Figure \ref{fig:Tvacuum}. Thus, the additional term in (\ref{TQ}) with $%
h_{0}(\mathcal{A})=1$ is due to a \textquotedblleft vacuum
mode\textquotedblright, a path which winds around the cylinder. }

\textrm{Let us now set $k=1$ and consider the state $\hat{\omega}%
_{i}=(0,\ldots ,0,\underset{i}{1},0,\ldots 0) $ with $i<n$. Then $Q(u)\hat{%
\omega}_{i}=\hat{\omega}_{i}+u\hat{\omega}_{i+1} $ as can be deduced
graphically from (\ref{Qdef}). In order to understand the origin of the
additional term in (\ref{TQ}) which includes the factor $z$, it suffices to
look at the contribution from the $D$ operator, since $T=A+zD$ and $A,D$ do
not contain $z$. The action of $D$ can be easily computed using Figure \ref%
{fig:Tweights} and remembering that only row configurations contribute where
the outer edges are occupied. One finds that all row configurations cancel
except for one term which is depicted in Figure \ref{fig:TQk1}: again the
additional term in (\ref{TQ}) corresponds to a diagram where one path winds
around the cylinder. }
\end{exa}

\subsection{Quantum Hamiltonian and particle picture}

So far we discussed a statistical mechanics model whose transfer matrix can
be identified with the generating function of the affine plactic complete
symmetric polynomials. Similarly, we can link the generating function of the
affine plactic elementary symmetric polynomials, the auxiliary matrix $T$,
to a physical model in quantum mechancis. Interpret the Dynkin labels $m_{i}$
as occupation numbers of a site $i$ of a circular lattice - the Dynkin
diagram of $\widehat{\mathfrak{su}}(n)$ in Figure \ref{fig:dynkin} - and the
maps (\ref{phi*}) and (\ref{phi}) as particle creation and annihilation
operators, respectively. Introducing the quantum Hamiltonian%
\begin{equation}
H=-\frac{e_{1}(\mathcal{A})+e_{n-1}(\mathcal{A})}{2}=-\frac{1}{2}%
\sum_{j=1}^{n}(\varphi _{j+1}^{\ast }\varphi _{j}+z\varphi _{j}^{\ast
}\varphi _{j+1})
\end{equation}%
with $\varphi_{n+1}=z^{-1}\varphi_1,\;\varphi_{n+1}^\ast=z\varphi_1^\ast$,
as well as the conserved charges $H_{r}^{\pm }=-(e_{r}(\mathcal{A})\pm
e_{n-r}(\mathcal{A}))/2$ which because of \eqref{ncehcomm} are in
involution, $[H,H_{r}^{\pm }]=[H_{r}^{\pm },H_{r^{\prime }}^{\pm }]=0$,
yields an alternative physical interpretation of the combinatorial
structures described here. This quantum system is known as \emph{phase model}%
, see \cite{BIK} and references therein.

\begin{figure}[tbp]
\begin{equation*}
\includegraphics[scale=0.8]{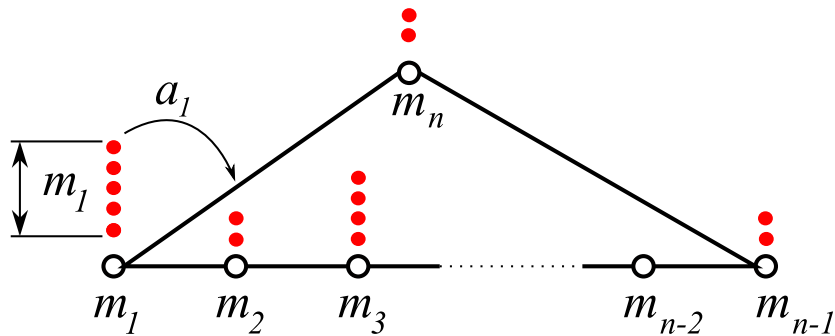}
\end{equation*}%
\caption{The fusion ring can be interpreted as a discrete model in quantum
mechanics. Each Dynkin label $m_i$ specifies a number of particles sitting
at the $i^{\text{th}}$ node of the Dynkin diagram. Application of the
generator $a_i$ of the affine plactic algebra, shifts one particle from site
$i$ to site $i+1$, as indicated in the picture for $i=1$.}
\label{fig:dynkin}
\end{figure}

\section{Bethe ansatz equations and the fusion potential}

Within the framework of exactly solvable models the next step is to
construct the eigenvectors of the transfer matrix $Q$. Instead it is simpler
to consider the eigenvalue problem of the auxiliary matrix $T$, since it
follows from the functional relation (\ref{TQ}) and the determinant formulae
(\ref{ncJT}) that the eigenvectors of $Q$ coincide with the eigenvectors of $%
T$. The advantage of this approach is that the eigenvectors of $T$ can be
computed via the algebraic Bethe ansatz or quantum inverse scattering method
(see e.g. \cite{KIB} for a text book and references therein) using the
commutation relations of the $A,B,C,D$ algebra in (\ref{momT}), which are
drastically simpler than the commutation relations of the algebra generated
by the matrix elements (\ref{momQ}). Starting point is a particular
assumption on the algebraic form of the eigenvectors: define an (off-shell)
Bethe vector at level $k>0$ to be%
\begin{equation}
b(x_{1},\ldots ,x_{k}):=B(x_{1}^{-1})\cdots B(x_{k}^{-1})\hat{\Omega},
\label{bvector}
\end{equation}%
where $\hat{\Omega}$ is again the unique vector corresponding to the
composition $(0,0,\ldots ,0)$. The requirement that (\ref{bvector}) is an
eigenvector of $T$ leads via the commutation relations of the Yang-Baxter
algebra contained in (\ref{RTT}) and a standard computation - which we omit
- to the Bethe ansatz equations \cite{BIK} \cite{KS}%
\begin{equation}
x_{1}^{n+k}=\cdots =x_{k}^{n+k}=z(-1)^{k-1}\prod_{i=1}^{k}x_{i}\;.
\label{BAE}
\end{equation}

We now discuss how the Bethe ansatz equations lead to a combinatorial
computation of fusion coefficients. We wish to emphasize that this is
possible without solving \eqref{BAE} first. For this reason we postpone the
discussion of their solutions to the next section, however, we already
mention that in order to solve them one needs to assume that $z^{\pm 1/n}$
exist.

The Bethe ansatz equations are $k$ polynomial equations and, thus, describe
an affine variety $\mathbb{V}_{n,k}^{\prime }\subset \mathbb{C}[z^{\pm \frac{%
1}{n}}]^{k}$. Recall that given a field $\mathbb{K}$ an affine variety is
usually defined in terms of a set of polynomials $f_{1},\ldots ,f_{k}\in
\mathbb{K}[x_{1},\ldots ,x_{k}]$ by setting $\mathbb{V}(f_{1},\ldots
,f_{k}):=\left\{ v\in \mathbb{K}^{k}:f_{i}(v_{1},\ldots ,v_{k})=0,\text{ }%
i=1,\ldots ,k\right\} $. Note that due to the commutation relations $%
B(x)B(y)=B(y)B(x),$ which follow from the Yang-Baxter equation (\ref{RTT}),
we can identify solutions of (\ref{BAE}) under permutations $(x_{1},\ldots
,x_{k})\sim (x_{w_{1}},\ldots ,x_{w_{k}})$ for all $w\in \mathbb{S}_{k},$
since they give rise to the same eigenvector (\ref{bvector}). Denote by $%
\mathbb{V}_{n,k}$ the variety obtained under this identification, then we
can think of $\mathbb{V}_{n,k}$ as being defined by $k$ elements in the ring
of \emph{symmetric} polynomials $\mathbb{C}[z^{\pm \frac{1}{n}%
}]^{k}[x_{1},\ldots ,x_{k}]^{\mathbb{S}_{k}}\cong \mathbb{C}[z^{\pm \frac{1}{%
n}}]^{k}[e_{1},\ldots ,e_{k}]$, where the $e_{i}$'s are the \emph{elementary
symmetric functions}.

\begin{lmm}
Let $h_{r}=\det (e_{1-i+j})_{1\leq i,j\leq r}$, the complete symmetric
functions, then the affine variety $\mathbb{V}_{n,k}$ defined by the Bethe
ansatz equations (up to permutations of the solutions) is given by\footnote{%
In \cite{KS} an additional relation has been stated to facilitate the
comparison with the small quantum cohomology ring of the Grassmannian, $%
h_{n+k}=z(-1)^{k-1}e_{k}$. This last relation follows from $%
h_{n}-z=h_{n+1}=\cdots =h_{n+k-1}=0$ by exploiting the definition the hook
Schur polynomial $s_{(n|k-1)}=e_{k}h_{n}=ze_{k}$; see (\ref{hooks}) below.}%
\begin{equation}
\mathbb{V}_{n,k}=\mathbb{V}(h_{n}-z,h_{n+1},\ldots ,h_{n+k-1})\;.
\label{bethe variety}
\end{equation}
\end{lmm}

The proof of this statement can be found in \cite[Lemma 6.3]{KS} and simply
uses the recursive relation (\ref{hrec}) for complete symmetric polynomials.
Next we assign to the solutions of the Bethe ansatz equations (\ref{BAE}) an
ideal in the ring of symmetric functions. Recall the definition of the
nullstellen or vanishing ideal of an affine variety $\mathbb{V}$, $\mathcal{I%
}(\mathbb{V}):=\{f\in \mathbb{C}(z)[e_{1},\ldots ,e_{k}]:f(v)=0,\;v\in
\mathbb{V}\}$. Then we have the following statement \cite[Proof of Theorem
6.20]{KS}:

\begin{prp}
Let $\mathcal{I}(\mathbb{V}_{n,k})$ be the nullstellen ideal of the Bethe
ansatz variety (\ref{bethe variety}), then
\begin{equation}
\langle h_{n}-z,h_{n+1},\ldots ,h_{n+k-1}\rangle =\mathcal{I}(\mathbb{V}%
_{n,k})\;.  \label{bethe ideal}
\end{equation}
\end{prp}

This proposition is proved via Hilbert's Nullstellensatz which asserts that
given an ideal $I$ in a polynomial ring one has $\mathcal{I}(\mathbb{V}(I))=%
\sqrt{I},$ where $\sqrt{I}=\{f:f^{m}\in I$ for some $m\in \mathbb{Z}_{>0}\}$
is the radical of $I$. In the present case where $I=\langle
h_{n}-1,h_{n+1},\ldots ,h_{n+k-1}\rangle $ one shows that $I=\sqrt{I}$ and
thus the assertion follows from the previous lemma.

\begin{rem}
\textrm{For $z=1$ the ideal (\ref{bethe ideal}) can be encoded into a fusion
potential $V'_{k+n}=p_{k+n}/(k+n)+(-1)^ke_k$ with  $p_{k+n}=\sum_{i=1}^{k}x_{i}^{n+k}$, the $(k+n)^{\text{th}}$ power
sum, noting that%
\begin{equation*}
\frac{1}{k+n}\frac{\partial p_{k+n}}{\partial e_{r}}=(-1)^{r-1}h_{k+n-r}=0,%
\qquad r=1,2,\ldots ,k\;.
\end{equation*}%
This is very similar to the fusion potential introduced by Gepner. The
difference lies in the constraints imposed on the variables: Gepner's fusion
potential \cite[Equation (2.31)]{Gepner}%
\begin{equation}
V_{n+k}=\frac{1}{k+n}\sum_{i=1}^{n}y_{i}^{n+k},  \label{gepner potential}
\end{equation}%
is defined in terms of $n$ variables $\{y_{1},\ldots ,y_{n}\}$ subject to
the constraints $e_{n}(y)=y_{1}\cdots y_{n}=1$ and
\begin{equation}
\frac{\partial V_{n+k}}{\partial e_{i}}=(-1)^{i+1}h_{n+k-i}=0,\quad
i=1,\ldots ,n-1\;.  \label{gepner ideal}
\end{equation}%
These constraints can be shown to be equivalent with the following set of
equations,%
\begin{equation}
y_{1}^{n+k}=\cdots =y_{n}^{n+k}=h_{k+n}(y_{1},\ldots
,y_{n})=(-1)^{n-1}h_{k}(y_{1},\ldots ,y_{n}),  \label{dual BAE}
\end{equation}%
which look very similar to the Bethe ansatz equations (\ref{BAE}). However,
the corresponding construction of the Bethe vector in terms of a Yang-Baxter
algebra is currently missing. }
\end{rem}

The importance of the ideal (\ref{bethe ideal}) derived form the Bethe
ansatz equations lies in the fact that it yields a presentation of the
Verlinde or fusion algebra in the ring of symmetric functions \cite[Theorem
6.20]{KS}.

\begin{thm}[Korff-Stroppel]
Set $z=1$. Then the map $P_{k}^{+}\ni \hat{\lambda}\mapsto s_{\lambda ^{t}}+%
\mathcal{I}(\mathbb{V}_{n,k}),$ where $\mathcal{I}(\mathbb{V}_{n,k})$ is the
ideal in (\ref{bethe ideal}) provides an algebra isomorphism,%
\begin{equation}
\mathcal{F}_{n,k}\otimes _{\mathbb{Z}}\mathbb{C}\cong \mathbb{C}%
[e_{1},\ldots ,e_{k}]/\mathcal{I}(\mathbb{V}_{n,k})\;.  \label{bethe iso}
\end{equation}
\end{thm}

In contrast Gepner's fusion potential is equivalent to a different
presentation; c.f. \cite[p247, result (4)]{GoodmanWenzl} and \cite[Equation
(2.36)]{Gepner}.

\begin{thm}[Gepner, Goodman-Wenzl]
The map $P_{k}^{+}\ni \hat{\lambda}\mapsto s_{\lambda }+I_{n,k},$ where $%
I_{n,k}=\langle e_{n}-1,h_{k+1},\ldots ,h_{n+k-1}\rangle $ is the ideal
resulting from (\ref{gepner ideal}) also provides an isomorphism,
\begin{equation}
\mathcal{F}_{n,k}\otimes _{\mathbb{Z}}\mathbb{C}\cong \mathbb{C}%
[e_{1},\ldots ,e_{n}]/I_{n,k}\;.  \label{gepner iso}
\end{equation}
\end{thm}

\begin{proof}
For the sake of completeness we briefly outline a proof of (\ref{gepner iso}%
) by showing that the ideal (\ref{gepner ideal}) following from Gepner's
fusion potential is identical with the ideal used by Goodman and Wenzl in
\cite{GoodmanWenzl} (with the extra condition $e_n=1$) who proved that the fusion ring is isomorphic to a
certain representation of the Hecke algebra at a primitive $(n+k)^{\text{th}}$
root of unity.

Let $J_{n,k}$ be the ideal generated from $e_{n}-1$ and the Schur
polynomials of the form $s_{(\lambda _{1},\lambda _{2},\ldots ,\lambda
_{n})} $ with $\lambda _{1}-\lambda _{n}=k+1$; c.f. \cite[p247, result (4)]%
{GoodmanWenzl}. Both ideals can be shown to be radical along similar lines
as it is discussed for (\ref{bethe ideal}) in \cite[Proof of Theorem 6.20,
Claim 1]{KS}. Hence, employing the
Nullstellensatz twice, $\mathcal{I}(\mathbb{V}(I_{n,k}))=I_{n,k}$ and $%
\mathcal{I}(\mathbb{V}(J_{n,k}))=J_{n,k}$, it suffices to prove the two
inclusions $J_{n,k}\subseteq \mathcal{I}(\mathbb{V}(I_{n,k}))$ and $%
I_{n,k}\subseteq \mathcal{I}(\mathbb{V}(J_{n,k}))$. For this purpose we
recall the definition of hook Schur polynomials \cite[Chapter I, Section 3, Example 9]{MacDonald}%
\begin{equation}
s_{(a|b)}:=s_{(a+1,1^{b})}=h_{a+1}e_{b}-h_{a+2}e_{b-1}+\cdots
+(-1)^{b}h_{a+b+1}\;.  \label{hooks}
\end{equation}%
Using the Frobenius notation $\lambda =(\alpha _{1},\ldots ,\alpha
_{r}|\beta _{1},\ldots ,\beta _{r})$ for a partition $\lambda $, where $%
\alpha _{i}$ and $\beta _{i}$ are the lengths of the horizontal and vertical
part of a hook centered at the $i^{\text{th}}$ box in the diagonal of $%
\lambda $, one has
\begin{equation}
s_{\lambda }=%\det (s_{(\alpha _{i}|\beta _{j})})_{1\leq i,j\leq r}=
\left\vert
\begin{array}{cccc}
s_{(\alpha _{1}|\beta _{1})} & s_{(\alpha _{1}|\beta _{2})} & \cdots &
s_{(\alpha _{1}|\beta _{r})} \\
s_{(\alpha _{2}|\beta _{1})} & \ddots &  & s_{(\alpha _{2}|\beta _{r})} \\
\vdots &  & \ddots & \vdots \\
s_{(\alpha _{r}|\beta _{1})} & s_{(\alpha _{r}|\beta _{2})} & \cdots &
s_{(\alpha _{r}|\beta _{r})}%
\end{array}%
\right\vert  \label{hookdet}
\end{equation}%
We first show that $J\subseteq \mathcal{I}(\mathbb{V}(I_{n,k}))=I_{n,k}.$
Exploiting $e_{n}=1$ and the Pieri rule $e_{n}s_{(\lambda _{1},\lambda
_{2},\ldots ,\lambda _{n})}=s_{(\lambda _{1}+1,\lambda _{2}+1,\ldots
,\lambda _{n}+1)}$ we can restrict ourselves to Schur polynomials of the
form $s_{(k+1,\lambda _{2},\ldots ,\lambda _{n-1},0)}$. From the definition
of $I_{n,k}$ it then follows that%
\begin{equation*}
s_{(k|b)}=h_{k+1}e_{b}-h_{k+2}e_{b+1}+\cdots +(-1)^{b}h_{k+b+1}=0
\end{equation*}%
for all $b=0,1,2,\ldots ,n-2$ and, thus, we can conclude with the help of (%
\ref{hookdet}) that $s_{(k+1,\lambda _{2},\ldots ,\lambda _{n-1},0)}=0$ for
all $k+1\geq \lambda _{2}\geq \cdots \geq \lambda _{n-1}$ as required.

The converse inclusion $I_{n,k}\subset \mathcal{I}(\mathbb{V}(J_{n,k}))$ is
easily derived along similar lines:%
\begin{eqnarray*}
0 &=&s_{(k|0)}=s_{(k+1,0,\ldots ,0)}=h_{k+1}, \\
0 &=&s_{(k|1)}=h_{k+1}e_{1}-h_{k+2}=-h_{k+2}, \\
&&\vdots \\
0 &=&s_{(k|n-2)}=h_{k+1}e_{n-2}-\cdots
+(-1)^{n-2}h_{k+n-1}=(-1)^{n-2}h_{k+n-1}\;.
\end{eqnarray*}%
This shows that the ideal of Goodman and Wenzel coincides with the ideal of
Gepner and using \cite[p247, result (4)]{GoodmanWenzl} we arrive at (\ref%
{gepner iso}).
\end{proof}

\begin{rem}
\textrm{The two isomorphisms (\ref{bethe iso}) and (\ref{gepner iso}) lead
to different expressions for the fusion coefficients in terms of
Littlewood-Richardson coefficients. We will discuss the case (\ref{bethe iso}%
) below. Goodman and Wenzl used their presentation (\ref{gepner iso}) to
derive the Kac-Walton formula \cite{Kac} \cite{Walton} (compare with \cite[%
p247, result (6)]{GoodmanWenzl}),%
\begin{equation}
\mathcal{N}_{\hat{\lambda}\hat{\mu}}^{(k)\hat{\nu}}=\sum_{w\in \hat{W}%
}\varepsilon (w)c_{\lambda \mu }^{(w\cdot \hat{\nu})^{\prime }},
\label{KWformula}
\end{equation}%
where $\hat{W}$ denotes the affine Weyl group, $\varepsilon(w)$ the
signature of $w$, $w\cdot \hat{\nu}=w(\hat{\nu}+\hat{\rho})-\hat{\rho}$ is
the shifted Weyl group action with $\hat{\rho}=\sum_{i}\hat{\omega}_{i}$
being the affine Weyl vector and $(w\cdot \hat{\nu})^{\prime }$ is the
partition obtained under the bijection (\ref{weight2part}). }
\end{rem}

\begin{exa}
\textrm{\label{KWex}Set $n=3$ and $k=4$ and consider the affine weights $%
\hat{\lambda}=\hat{\omega}_{0}+2\hat{\omega}_{1}+\hat{\omega}_{2}$, $\hat{\mu%
}=\hat{\omega}_{0}+\hat{\omega}_{1}+2\hat{\omega}_{2}$ in $P_{k}^{+}$. The
corresponding partitions under \eqref{weight2part} are $\lambda =(3,1)$ and $%
\mu =(3,2)$. Employing the Littlewood-Richardson rule (see e.g. \cite%
{FultonYT}) yields the partitions%
\begin{multline}
\rho =(6,3,0),(6,2,1),(5,4,0),(5,3,1),(5,3,1),(5,2,2),(4,4,1), \\
(4,3,2),(4,3,2),(3,3,3),(5,2,1,1),(4,3,1,1),(4,2,2,1),(3,3,2,1)\;.
\label{LRRex}
\end{multline}%
Discarding all partitions of length \TEXTsymbol{>} $n$ and removing all $n$%
-columns from the partitions $\nu $ we obtain the following $\mathfrak{su}%
(n) $ tensor product decomposition%
\begin{equation*}
\lambda \otimes \mu =(6,3)\oplus (5,1)\oplus (5,4)\oplus 2(4,2)\oplus
(3,0)\oplus (3,3)\oplus 2(2,1)\oplus (0,0)\;.
\end{equation*}%
Here we have identified highest weight modules with the corresponding
partitions. We wish to consider the fusion coefficient of the affine weight $%
\hat{\nu}=2\hat{\omega}_{1}+2\hat{\omega}_{2}$ with partition $\nu =(4,2)$.
From (\ref{KWformula}) we find%
\begin{equation}
\mathcal{N}_{\hat{\lambda}\hat{\mu}}^{(k)\hat{\nu}}=c_{\lambda \mu
}^{(4,2)}-c_{\lambda \mu }^{(6,3)}=2-1=1,  \label{KW2LR}
\end{equation}%
since $s_{0}\cdot \widehat{(6,3)}=s_{0}\cdot (-2\hat{\omega}_{0}+3\hat{\omega%
}_{1}+3\hat{\omega}_{2})=\hat{\nu}$ with $s_{0}$ denoting the affine Weyl
reflection. In fact, the entire fusion product expansion is computed to%
\begin{equation}
\Yvcentermath1\yng(3,1)~\ast ~\yng(3,2)=  \label{fusionex3} \\
\Yvcentermath1\yng(4,2)+~\yng(3)+~\yng(3,3)+2~\yng(2,1)+\emptyset \;.
\end{equation}%
}
\end{exa}

\subsection{Algorithm to compute fusion coefficients}

Based on the presentation (\ref{bethe iso}) derived from the Bethe ansatz
equations we now formulate an alternative algorithm how to compute fusion
coefficients in terms of Littlewood-Richardson numbers.

\begin{enumerate}
\item Compute the expansion $s_{\lambda ^{t}}s_{\mu ^{t}}=\sum_{\rho
^{t}}c_{\lambda \mu }^{\rho }s_{\rho ^{t}}$ via the Littlewood-Richardson
rule ; note that $c_{\lambda \mu }^{\rho }=c_{\lambda ^{t}\mu ^{t}}^{\rho
^{t}}$ \cite{FultonYT}. Discard all terms for which the partition $\rho ^{t}$
has length \TEXTsymbol{>} $k$.

\item For each of the remaining terms with $\rho _{1}^{t}\geq n$ make the
replacement $s_{\rho ^{t}}=s_{(\rho _{2}^{t},\ldots ,\rho _{k}^{t},\rho
_{1}^{t}-n)}$. Whenever $(\rho _{2}^{t},\ldots ,\rho _{k}^{t},\rho
_{1}^{t}-n)$ is not a partition use the straightening rules \cite{MacDonald}
\begin{equation*}
s_{(\ldots ,a,b,\ldots )}=-s_{(\ldots ,b-1,a+1,\ldots )}\quad \text{and}%
\quad s_{(\ldots ,a,a+1,\ldots )}=0
\end{equation*}%
for Schur polynomials to rewrite $s_{(\rho _{2}^{t},\ldots ,\rho
_{k}^{t},\rho _{1}^{t}-n)}$ as $s_{\nu ^{t}}$ with $\nu ^{t}$ a partition.

\item Collecting terms for each $\nu $ one obtains the fusion coefficient $%
\mathcal{N}_{\hat{\lambda}\hat{\mu}}^{(k)\hat{\nu}}$.
\end{enumerate}

\begin{exa}
\textrm{\label{rimhookex2}As in example (\ref{KWex}) set $n=3,\,k=4$ and
consider the partitions $\lambda =(3,1,0)$ and $\mu =(3,2,0).$ After taking
the transpose partitions in (\ref{LRRex}) we discard $\rho
^{t}=(2,2,2,1,1,1) $, $(3,2,1,1,1,1)$, $(2,2,2,2,1)$, $(3,2,2,1,1)$, $%
(3,3,1,1,1)$ and $(4,2,1,1,1)$ as they have length \TEXTsymbol{>} $k=4$. We
are left with three partitions $\rho ^{t}$ for which $\rho _{1}^{t}>n$,
namely $(4,2,2,1)$, $(4,3,1,1)$, $(4,3,2,0)$. Employing the above algorithm
we calculate%
\begin{equation*}
s_{(4,2,2,1)}=s_{(2,2,1,1)},\quad s_{(4,3,1,1)}=s_{(3,1,1,1)},\quad
s_{(4,3,2,0)}=s_{(3,2,0,1)}=0\;.
\end{equation*}%
Removing all rows of length $n=3$ and collecting terms, one arrives after
taking the transpose again at the expansion (\ref{fusionex3}), however, the
expression for the fusion coefficients in terms of Littlewood-Richardson
numbers are different. For instance, the coefficient of $\nu =(4,2,0)$ is
according to the Kac-Walton formula the difference of two
Littlewood-Richardson coefficients, see (\ref{KW2LR}). In contrast, here we
find that%
\begin{equation*}
\mathcal{N}%
_{(3,1,0)(3,2,0)}^{(k)(4,2,0)}=c_{(3,1,0,0),(3,2,0,0)}^{(4,3,1,1)}=1
\end{equation*}%
Thus, the Bethe ansatz equations (\ref{BAE}) provide an alternative
algorithm to compute fusion coefficients. }
\end{exa}

\begin{proof}[Proof of the algorithm]
Consider the ring of symmetric functions $\mathbb{C}%
[x_{1},x_{2},\ldots ,x_{k}]^{\mathbb{S}_{k}}$, then $s_{\lambda }=0$ unless $%
\ell (\lambda )\leq k$. This justifies Step 1 of the algorithm. To deduce Step 2, assume we
are given a partition $\lambda $ with $\ell (\lambda )\leq k$ and $\lambda
_{1}>n$. Then we rewrite the Schur function as \cite{MacDonald}%
\begin{equation*}
s_{\lambda }(x_{1},x_{2},\ldots ,x_{k})=\sum_{w\in \mathbb{S}_{k}}w\cdot
\left( x_{1}^{\lambda _{1}}\cdots x_{k}^{\lambda _{k}}\theta (x)\right)
,\quad \theta (x):=\prod_{1\leq i<j\leq k}\frac{1}{1-x_{j}/x_{i}}\;.
\end{equation*}%
Observing that the equations (\ref{BAE}) for $x_{1}$ are equivalent to%
\begin{equation*}
x_{1}^{n}\prod_{j>1}\frac{x_{1}}{x_{1}-x_{j}}=z\prod_{j>1}\frac{x_{j}}{%
x_{j}-x_{1}}
\end{equation*}%
one derives the identity
\begin{equation*}
x_{1}^{\lambda _{1}}\cdots x_{k}^{\lambda _{k}}\theta (x)=z~w_{0}\cdot
(x_{1}^{\lambda _{2}}x_{2}^{\lambda _{3}}\cdots x_{k}^{\lambda _{1}-n}\theta
(x)),
\end{equation*}%
where $w_{0}=\sigma _{k-1}\cdots \sigma _{2}\sigma _{1}$ and $\sigma _{i}$
is the transposition which permutes $x_{i}$ and $x_{i+1}$. Insertion of this
identity into the above expression for the Schur function proves that $%
s_{\lambda }=s_{(\lambda _{2},\lambda _{3},\ldots ,\lambda _{k},\lambda
_{1}-n)}+\mathcal{I}(\mathbb{V}_{n,k})$. Step 3 then follows from (\ref%
{bethe iso}).
\end{proof}

\section{Bethe vectors as idempotents}

We now solve the Bethe ansatz equations (\ref{BAE}) explicitly, which is
possible due to their simple form, and describe the variety $\mathbb{V}%
_{n,k} $ which consists of a discrete set of points in $\mathbb{C}%
[z^{1/n}]^{k}$. For each partition $\sigma $ in $\mathcal{P}_{\leq n-1,k}$
define the following tuple of elements in $\mathbb{C}[z^{\pm 1/n}],$%
\begin{equation}
\mathcal{P}_{\leq n-1,k}\ni \sigma \mapsto x^{\sigma }=z^{\frac{1}{n}}\zeta
^{\frac{|\sigma |}{n}}(\zeta ^{I_{1}(\sigma ^{t})},\ldots ,\zeta
^{I_{k}(\sigma ^{t})}),  \label{broots}
\end{equation}%
where $\zeta =\exp (\frac{2\pi i}{k+n})$ and the exponents are the
half-integers,%
\begin{equation}
I(\sigma ^{t})=\left( \tfrac{k+1}{2}+\sigma _{k}^{t}-k,\ldots ,\tfrac{k+1}{2}%
+\sigma _{1}^{t}-1\right) \;.  \label{Itmap}
\end{equation}%
A straightforward computation shows that $x^{\sigma }$ solves (\ref{BAE})
for any $\sigma $ \cite[Theorem 6.4]{KS}.

\begin{thm}[completeness of the Bethe ansatz]
Fix $k\geq 0$. Then the set of vectors $\{b_{\sigma }:=b(x_{1}^{\sigma
},\ldots ,x_{k}^{\sigma })\}_{\sigma \in \mathcal{P}_{\leq n-1,k}}$ defined
in terms (\ref{bvector}) and (\ref{broots}) forms an eigenbasis of the
transfer (\ref{Qdef}) and auxiliary matrix (\ref{T}) in the subspace $%
\mathbb{C}P_{k}^{+}$. In particular, let $S$ denote the modular $\widehat{%
\mathfrak{su}}(n)_{k}$ $S$-matrix (\ref{modS}), then the Bethe vector $%
b_{\sigma }$ has the expansion (for simplicity we now label $S$-matrix
elements with partitions)%
\begin{equation}
b_{\sigma }=z^{-k}\sum_{\hat{\lambda}\in P_{k}^{+}}z^{\lambda _{1}-\frac{%
|\lambda |}{n}}\frac{S_{\sigma ^{\ast }\lambda }}{S_{\sigma ^{\ast
}\emptyset }}~\hat{\lambda}\;.  \label{onshellbvector}
\end{equation}%
Moreover, one has the following eigenvalue equations for the affine plactic
polynomials (\ref{nch}) and (\ref{nce}),%
\begin{equation}
h_{r}(\mathcal{A})b_{\sigma }=z^{k-r+\frac{r}{n}}\frac{S_{(r)\sigma }}{%
S_{\emptyset \sigma }}~b_{\sigma }\qquad \text{and}\qquad e_{r}(\mathcal{A}%
)b_{\sigma }=z^{k-r+\frac{r}{n}}\frac{S_{(1^{r})\sigma }}{S_{\emptyset
\sigma }}~b_{\sigma }\;,  \label{spectrum}
\end{equation}%
where $(r)$ and $(1^{r})$ denote the partitions whose Young diagrams consist
respectively of a single row and a single column of length $r$.
\end{thm}

Inherent in the last result is the statement that the modular S-matrix can
be computed in terms of scalar products of the on-shell Bethe vectors $%
b_{\sigma }$ and, hence, ultimately in terms of the Yang-Baxter algebra
generator $B$ via (\ref{bvector}). Namely, from (\ref{onshellbvector}) we
obtain for $z=1$ that $\left\langle b_{\sigma },\lambda \right\rangle
=S_{\lambda \sigma }/S_{\emptyset \sigma }$ and $\left\langle b_{\sigma
},b_{\sigma }\right\rangle =|S_{\emptyset \sigma }|^{-2}$. Using that $%
S_{\emptyset \sigma }>0$ we find%
\begin{equation}
z=1:\qquad S_{\lambda \sigma }=\frac{\left\langle b_{\sigma },\lambda
\right\rangle }{\left\langle b_{\sigma },b_{\sigma }\right\rangle ^{\frac{1}{%
2}}},\qquad b_{\sigma }=B(\bar{x}_{1}^{\sigma })\cdots B(\bar{x}_{k}^{\sigma
})\hat{\Omega}\;.
\end{equation}%
Note in particular, that for $\sigma =\emptyset $ we obtain the groundstate $%
b_{\emptyset }=\sum_{\hat{\lambda}\in P_{k}^{+}}\bar{S}_{\emptyset \lambda
}/S_{\emptyset \emptyset }~\lambda $ of the quantum Hamiltonian (\ref{H}),
or equivalently the Perron-Frobenius eigenvector of the transfer matrix (\ref%
{Qdef}), whose components are given by the so-called \emph{quantum dimensions%
}%
\begin{equation}
\frac{S_{\lambda \emptyset }}{S_{\emptyset \emptyset }}=s_{\lambda
^{t}}(x^{\emptyset })=\prod_{\alpha >0}\frac{\zeta ^{\frac{\left\langle
\alpha ,\rho +\lambda \right\rangle }{2}}-\zeta ^{-\frac{\left\langle \alpha
,\rho +\lambda \right\rangle }{2}}}{\zeta ^{\frac{\left\langle \alpha ,\rho
\right\rangle }{2}}-\zeta ^{-\frac{\left\langle \alpha ,\rho \right\rangle }{%
2}}},\quad \zeta =e^{-\frac{2\pi i}{k+n}}\;.
\end{equation}%
Here the product runs over all positive roots of $\mathfrak{su}(n)$ and $%
\rho =\frac{1}{2}\sum_{\alpha >0}\alpha $ is the Weyl vector. The expression
in terms of Schur functions generalises to the excited states, we have in
general that $S_{\lambda \sigma }/S_{\emptyset \sigma }=s_{\lambda
^{t}}(x^{\sigma })$ which can be interpreted as characters evaluated at
special points; see \cite{KS} for details.

Since we have skipped the algebraic Bethe ansatz computation for $T$ let us
verify for the simple case $k=1$ that the Bethe vector (\ref{bvector}) is
indeed an eigenvector of the transfer matrix subject to the Bethe ansatz
equations (\ref{BAE}) and compute the modular S-matrix.

\begin{exa}
\textrm{Take $n>2$ and set $k=z=1$. Then it follows from (\ref{B&Q}) that%
\begin{multline*}
Q(u)B(v)\hat{\Omega}=\left( 1+\frac{u}{v}\right) B(v)\boldsymbol{Q}_{0,0}(u)%
\hat{\Omega} \\
+\sum_{r\geq 0}\boldsymbol{Q}_{r,r+1}(u)(D(v)-A(v))\hat{\Omega}+\sum_{r\geq
0}\boldsymbol{Q}_{r,r+2}(u)C(v)\hat{\Omega}
\end{multline*}%
Exploiting that $C(v)\hat{\Omega}=0$ and $\boldsymbol{Q}_{0,0}(u)\hat{\Omega}%
=\hat{\Omega}$, we need to choose the variable $v$ such that the second
summand on the right hand side vanishes. Observing that $(D(v)-A(v))\hat{%
\Omega}=(v^{n}-1)\hat{\Omega}$, we arrive at the Bethe ansatz equations (\ref%
{BAE}) with $x_{1}=v^{-1}$. The solutions are easily found to be $%
x_{1}(s)=\zeta ^{s/n}\zeta ^{s}=e^{-\frac{2\pi i}{n}s}$ with $s=0,1,\ldots
,n-1$ and the Bethe vector thus reads%
\begin{equation*}
b_{s}:=B(x_{1}(s)^{-1})\hat{\Omega}=\sum_{r=0}^{n-1}e^{\frac{2\pi i}{n}%
rs}~(1^{r})\;.
\end{equation*}%
The modular S-matrix is then easily computed to be%
\begin{equation*}
S_{(1^{r})(1^{s})}=\frac{e^{\frac{2\pi i}{n}rs}}{\sqrt{n}}\;.
\end{equation*}
}
\end{exa}

Given the eigenvalue equations (\ref{spectrum}) it is natural to define
\emph{affine plactic Schur polynomials} via the familiar Jacobi-Trudi
formula (we exploit once more the integrability condition (\ref{ncehcomm})
which guarantees that the determinant is well-defined),%
\begin{equation}
s_{\lambda }(\mathcal{A}):=\det (h_{\lambda _{i}-i+j}(\mathcal{A}))_{1\leq
i,j\leq n-1}\;.  \label{ncschur}
\end{equation}%
It is then not difficult to show that the affine plactic Schur polynomials
satisfy the eigenvalue equation $s_{\lambda }(\mathcal{A})b_{\sigma
}=(S_{\lambda \sigma }/S_{\emptyset \sigma })b_{\sigma }$ which leads to the
next result \cite[Proposition 6.11 and Theorem 6.12]{KS}.

\begin{crl}[combinatorial product]
Introduce a product on the subspace $\mathbb{C}P_{k}^{+}$ by setting%
\begin{equation}
\hat{\lambda}\ast \hat{\mu}:=s_{\lambda }(\mathcal{A})\hat{\mu},\qquad
\forall \hat{\lambda},\hat{\mu}\in P_{k}^{+}\;.  \label{product}
\end{equation}%
Then for $z=1$ $(\mathbb{C}P_{k}^{+},\ast )$ is a unital, associative and
commutative\footnote{%
For arbitrary $z$ the product is still associative but ceases to be
commutative. This is different from \cite[Theorem 6.12, eqn (6.33)]{KS}
where the product was defined in terms of $s_{\hat{\lambda}}(\mathcal{A})$
with $\hat{\lambda}$ being the partition obtained from $\lambda $ by adding $%
k-\lambda _{1}$ columns of height $n$. This introduces an additional $z$%
-dependence which renders the product commutative.} algebra isomorphic to
the fusion or Verlinde algebra $\mathcal{F}_{n,k}^{\mathbb{C}}$. Moreover,
the renormalised Bethe vectors $\hat{b}_{\lambda }=b_{\lambda }/\langle
b_{\lambda },b_{\lambda }\rangle ^{\frac{1}{2}}$ are idempotents with
respect to this product, i.e.
\begin{equation}
\hat{b}_{\lambda }\ast \hat{b}_{\mu }=\delta _{\lambda ,\mu }\hat{b}%
_{\lambda },\qquad \forall \lambda ,\mu \in \mathcal{P}_{\leq n-1,k}\;.
\end{equation}%
Thus, the completeness of the Bethe ansatz is equivalent to the
semi-simplicity of the fusion algebra.
\end{crl}

\begin{proof}
The proof of the result (\ref{product}) can be given in one line, hence we
repeat it here for the reader's convenience,%
\begin{equation*}
s_{\lambda }(\mathcal{A})\mu =\sum_{\sigma }\frac{\left\langle b_{\sigma
},\mu \right\rangle }{\left\langle b_{\sigma },b_{\sigma }\right\rangle }%
~s_{\lambda }(\mathcal{A})b_{\sigma }=\sum_{\sigma }\frac{\bar{S}_{\emptyset
\sigma }S_{\mu \sigma }S_{\lambda \sigma }}{S_{\emptyset \sigma }}b_{\sigma
}=\sum_{\nu }\left( \sum_{\sigma }\frac{S_{\lambda \sigma }S_{\mu \sigma }%
\bar{S}_{\nu \sigma }}{S_{\emptyset \sigma }}\right) ~\nu \;.
\end{equation*}%
We already pointed out above that (\ref{onshellbvector}) implies that $%
\left\langle b_{\sigma },\lambda \right\rangle =S_{\lambda \sigma
}/S_{\emptyset \sigma }$ and $\left\langle b_{\sigma },b_{\sigma
}\right\rangle =|S_{\emptyset \sigma }|^{-2}$, whence $\hat{b}_{\sigma
}=S_{\emptyset \sigma }\sum_{\hat{\lambda}}\bar{S}_{\lambda \sigma }\hat{%
\lambda}$. It now follows that
\begin{equation*}
\hat{b}_{\rho }\ast \hat{b}_{\sigma }=S_{\emptyset \rho }\sum_{\lambda }\bar{%
S}_{\lambda \rho }s_{\lambda }(\mathcal{A})\hat{b}_{\sigma }=\frac{%
S_{\emptyset \rho }}{S_{\emptyset \sigma }}\sum_{\lambda }\bar{S}_{\lambda
\rho }S_{\lambda \sigma }\hat{b}_{\sigma }=\delta _{\sigma ,\rho }\hat{b}%
_{\sigma }\;,
\end{equation*}%
where in the last step we have used unitarity, $S\cdot S^{\ast }=1$, of the
modular S-matrix.
\end{proof}

\subsection{Fusion matrices as affine plactic Schur polynomials}

It is well-known that the fusion matrices $N_{\hat{\lambda}}^{(k)}:=(%
\mathcal{N}_{\hat{\lambda}\hat{\mu}}^{(k),\hat{\nu}})_{\hat{\mu},\hat{\nu}%
\in P_{k}^{+}}$ form a representation of the fusion ring. From the existence
of the eigenbasis (\ref{onshellbvector}) and (\ref{product}) one deduces the
next corollary which states that the affine plactic Schur polynomials (\ref%
{ncschur}), when restricted to the subspace $\mathbb{C}P_{k}^{+},$ are
identical with the fusion matrices.

\begin{crl}
\label{spectral curve}Denote by $\mathcal{F}_{n,k}^{\prime }\subset \limfunc{%
End}(\mathbb{C}P_{k}^{+})$ the subalgebra generated by the (restricted)
affine plactic Schur polynomials $\{s_{\lambda }(a)_{k}\}_{\lambda \in
\mathcal{P}_{\leq n-1,k}}$. The map $s_{\lambda }(a)_{k}\mapsto \lbrack
s_{\lambda ^{t}}]$ provides an isomorphism ($z=1$), $\mathcal{F}%
_{n,k}^{\prime }\cong \mathbb{C}[e_{1},\ldots ,e_{k}]/\mathcal{I}(\mathbb{V}%
_{n,k})$. In particular, one has%
\begin{equation}
s_{\lambda }(a)_{k}s_{\mu }(a)_{k}=\sum_{\nu \in \mathcal{P}_{\leq n-1,k}}%
\mathcal{N}_{\hat{\lambda}\hat{\mu}}^{(k),\hat{\nu}}s_{\nu }(a)_{k}\;.
\label{ncschurexp}
\end{equation}
\end{crl}

\begin{rem}
\textrm{It is common knowledge within the statistical mechanics community
that a set of commuting transfer matrices is the distinguishing property of
an integrable or exactly solvable lattice model \cite{Baxter}. Due to the
development of the quantum inverse scattering method by the Faddeev school,
it is the noncommutative structures, the Yang-Baxter algebras discussed in
Section 3, which have been the centre of attention. The result (\ref%
{ncschurexp}) shows that also the commutative (sub)algebra, the
\textquotedblleft integrals of motion\textquotedblright , have an
interesting structure with applications in representation and, here,
conformal field theory. }
\end{rem}

\section{Recursion formulae for fusion coefficients}

The quantum mechanical interpretation via the Hamiltonian (\ref{H})
identifies the fusion ring as the $k$-particle superselection sector of the
state space $\mathbb{C}P^{+}$. The physical picture of creating and
destroying particles via the maps (\ref{phi*}) and (\ref{phi}) suggests to
investigate how fusion coefficients at different levels are related.

Let us start from the simple observation that according to (\ref{A}) the
operator $A(u)$ does not depend on $a_{n}$. We thus find from (\ref{comm3})
that $A(u)\varphi _{1}=\varphi _{1}A(u)$ and, hence,
\begin{equation}
\varphi _{1}T(u)\varphi _{1}^{\ast }=\varphi _{1}A(u)\varphi _{1}^{\ast
}+z\varphi _{n}\varphi _{1}A(u)\varphi _{1}^{\ast }=T(u)\;.  \label{Trec}
\end{equation}%
Similarly, we can argue that $\boldsymbol{Q}_{0,0}(u)\varphi _{n}^{\ast
}=\varphi _{n}^{\ast }\boldsymbol{Q}_{0,0}(u)$ in (\ref{Q00}) and exploiting
once more (\ref{comm3}) we arrive with the help of (\ref{Qdef2}) at%
\begin{equation}
\varphi _{n}Q(u)\varphi _{n}^{\ast }=Q(u)\;.  \label{Qrec}
\end{equation}%
We stay in the particle picture and set $z=1$. Then the action of $h_{k}(%
\mathcal{A})|_{\mathbb{C}P_{k}^{+}}$ is particularly simple: all particles
on the $\widehat{\mathfrak{su}}(n)$ Dynkin diagram are shifted by one
position, $h_{k}(\mathcal{A})\boldsymbol{m}=(m_{n},m_{1},m_{2},\ldots
,m_{n-1})$. Because of (\ref{ncehcomm}), the above commutation relations (%
\ref{Trec}) and (\ref{Qrec}) generalise for $z=1$ to the maps $\varphi
_{i},\varphi _{i}^{\ast }$ with $1\leq i\leq n$. If we recall that $T$ and $%
Q $ are the generating functions of the affine plactic elementary and
complete symmetric polynomials we obtain immediately the following:

\begin{prp}[recursion formulae]
For any $\hat{\mu},\hat{\nu}\in P_{k}^{+}$ we have the identities%
\begin{equation}
\mathcal{N}_{(\widehat{1^{r}})\varphi _{i}^{\ast }\hat{\mu}}^{(k+1),\varphi
_{i}^{\ast }\hat{\nu}}=\mathcal{N}_{(\widehat{1^{r}})\hat{\mu}}^{(k),\hat{\nu%
}}\qquad \text{and}\qquad \mathcal{N}_{\widehat{(r)}\varphi _{i}^{\ast }\hat{%
\mu}}^{(k+1),\varphi _{i}^{\ast }\hat{\nu}}=\mathcal{N}_{\widehat{(r)}\hat{%
\mu}}^{(k),\hat{\nu}},  \label{fusionrec}
\end{equation}%
where $0\leq r\leq k$ and $i=1,2,\ldots ,n$.
\end{prp}

\begin{exa}
\textrm{Set $n=5$ and $\mu =(2,2,1,0)$. Then for $r=3$ we find with help of
the algorithm of Section 4.1 the following expansion at level $k=2,$%
\begin{equation*}
\Yvcentermath1\yng(1,1,1)~\ast ~\yng(2,2,1)=\Yvcentermath1\yng(1,1,1)+~\yng%
(2,1)\;.
\end{equation*}%
To verify the first identity in (\ref{fusionrec}) for $i=1$ observe that the
first two terms in the following expansion for $\varphi_1^\ast\mu=(3,2,1)$
at level $k=3$,%
\begin{equation*}
\Yvcentermath1\yng(1,1,1)~\ast ~\yng(3,2,1)=\Yvcentermath1\yng(2,1,1)+~\yng%
(3,1)+~\yng(2,2)+~\yng(3,3,2,1)\;,
\end{equation*}
are obtained from the $k=2$ expansion by adding a one-column.}
\end{exa}

Another set of relations follows from the recursion relations (\ref{Trec})
and (\ref{Qrec}). As we discussed earlier, the boundary parameter $z$ allows
us to project on the finite plactic algebra $\mathcal{A}^{\prime }$ by
setting $z=0$. With the help of (\ref{nce0action}), (\ref{nch0action}) and (%
\ref{product}) one proves the following statement.

\begin{prp}
Let $\hat{\mu},\hat{\nu}\in P_{k}^{+}$ and $\mu ,\nu $ the corresponding
partitions under (\ref{weight2part}), then we have for $r=0,1,\ldots ,n-1$
the following expression in terms of Littlewood-Richardson numbers%
%\begin{equation*}
%\mathcal{N}_{(\widehat{1^{r}})\hat{\mu}}^{(k),\hat{\nu}}=\delta _{\mu
%_{1},\nu _{1}}c_{\mu ~(1^{r})}^{\nu }+\delta _{\mu _{1}+1,\nu
%_{1}}c_{\varphi _{1}^{\ast }\mu ~(1^{r-1})}^{\varphi _{n}^{\ast }\nu }\;.
%\end{equation*}%
\begin{equation}
\mathcal{N}_{(\widehat{1^{r}})\hat{\mu}}^{(k),\hat{\nu}}=\left\{
\begin{array}{cc}
1, &
\begin{array}{c}
\text{if either }\mu _{1}=\nu _{1}\text{ and }\nu /\mu =(1^{r}) \\
\text{or }\mu _{1}+1=\nu _{1}\text{ and }\varphi _{n}^{\ast }\nu /\varphi
_{1}^{\ast }\mu =(1^{r-1})%
\end{array}
\\
0, & \text{else}%
\end{array}%
\right. \;.  \label{FusionLRv}
\end{equation}%
In case of an horizontal strip of length $r=0,1,\ldots ,k$ one has instead
the recursion relation%
%\begin{equation*}
%\mathcal{N}_{\widehat{(r)}\hat{\mu}}^{(k),\hat{\nu}}=\delta _{\mu _{1},\nu
%_{1}}c_{\mu (r)}^{\nu }+\mathcal{N}_{\widehat{(r-1)}~\varphi _{n}\hat{\mu}%
%}^{(k-1),\varphi _{1}\hat{\nu}}\;.
%\end{equation*}%
\begin{equation}
\mathcal{N}_{\widehat{(r)}\hat{\mu}}^{(k),\hat{\nu}}=\left\{
\begin{array}{cc}
1, & \text{if }\mu _{1}=\nu _{1}\text{ and }\nu /\mu =(r) \\
\mathcal{N}_{\widehat{(r-1)}~\varphi _{n}\hat{\mu}}^{(k-1),\varphi _{1}\hat{%
\nu}}, & \text{else}%
\end{array}%
\right. \;.  \label{FusionLRh}
\end{equation}
\end{prp}

\begin{exa}
\textrm{Again we set $n=5$ and $\mu =(2,2,1,0)$. Then at level $k=4$ one
computes via the above algorithm the expansion%
\begin{equation*}
\Yvcentermath1\yng(3)~\ast ~\yng(2,2,1)=\Yvcentermath1\yng(3,2,2,1)+~\yng%
(4,2,1,1)+~\yng(4,2,2)\;.
\end{equation*}%
Since for all partitions $\nu $ appearing in the expansion one has $\nu
_{1}>\mu _{1}$, only the second case in (\ref{FusionLRh}) applies and, thus,
we find the following nonzero fusion coefficients at level $k=3$,%
\begin{equation*}
\mathcal{N}_{(2)~(2,2,1)}^{(3),(2,2,2,1)}=\mathcal{N}%
_{(2)~(2,2,1)}^{(3),(3,2,1,1)}=\mathcal{N}_{(2)~(2,2,1)}^{(3),(3,2,2)}=1\;.
\end{equation*}%
The latter can be verified by using again the presentation (\ref{bethe iso})
in the ring of symmetric functions and the resulting algorithm or the
Verlinde formula. }
\end{exa}

\begin{rem}
\textrm{Setting $z=0$ in (\ref{bethe ideal}) it appears at first sight that
we should expect to obtain the cohomology ring of the Grassmannian,
\begin{equation*}
H^{\ast }(Gr_{n-1,n+k-1})\cong \mathbb{Z}[e_{1},\ldots ,e_{n}]/\left\langle
h_{n},\ldots ,h_{n+k-1}\right\rangle \;.
\end{equation*}%
However, because of the $z$-dependence entering the construction of the
eigenvectors (\ref{onshellbvector}) one cannot conclude that the product (%
\ref{product}) specialises to the cup product in $H^{\ast }(Gr_{n-1,n+k-1})$%
. Hence, the structure constants of the algebra defined via (\ref{product})
become in this special case Littlewood-Richardson coefficients with the
additional constraint that $\mu _{1}=\nu _{1}$. }
\end{rem}

\section{Conclusions}

While we have focussed here on a particularly simple integrable model and a
very special ring, the observations made should generalise to a wider class
of exactly solvable lattice models. Let us outline the general concept.

Starting point in the construction of exactly solvable lattice models are in
general solutions to the Yang-Baxter, or more generally, the star-triangle
equation. The overwhelming majority of these solutions can be constructed
with the help of representations of some \emph{noncommutative} algebras,
such as $q$-deformed enveloping algebras of Kac-Moody algebras
(Drinfel'd-Jimbo \textquotedblleft quantum groups\textquotedblright ) or
elliptic generalisations thereof. For the example at hand the noncommutative
algebras in question are the phase and affine plactic algebra and it has
been explained in \cite{Korff} how these have their origin in the $q$%
-deformed enveloping algebra $U_{q}\widehat{\mathfrak{sl}}(2)$. Once the
solutions are known explicitly one can interpret their matrix elements as
the Boltzmann weights of a statistical vertex model - depending on a free
parameter - and construct the corresponding transfer matrices to compute its
partition function. Because the Boltzmann weights satisfy the Yang-Baxter
equation the transfer matrices commute, thus defining a \emph{commutative}
(and associative) algebra or ring despite being built from the generators of
a \emph{noncommutative} algebra. We have seen how the transfer matrices are
generating functions for polynomials in the alphabet $(a_{1},\ldots ,a_{n})$
and while the letters $a_{i}$ of this alphabet do not commute the affine
plactic Schur polynomials (\ref{ncschur}) do.

Via the Bethe ansatz one then computes the idempotents of this commutative
algebra, showing that it is semi-simple, and also its structure constants,
the fusion coefficients and their expression in terms of the Verlinde
formula. While the algebraic Bethe ansatz employed here is special to the $%
U_{q}\widehat{\mathfrak{sl}}(2)$ case, generalisations of it which are
applicable to higher rank, such as the nested, analytic or coordinate Bethe
ansatz might be used instead. It is true for the majority of models solvable
by the Bethe ansatz that the Bethe ansatz equations determining the
eigenvectors of the transfer matrix (or the idempotents of the commutative
algebra) are given in terms of some elements in a commutative polynomial
ring $\mathcal{R}$ and hence define an affine variety $\mathbb{V}$. In the
example discussed here the polynomial ring $\mathcal{R}$ has been the the
ring of symmetric functions $\mathcal{R}=\mathbb{C}[e_{1},\ldots ,e_{k}]$.
In the next step one determines the nullstellen ideal $\mathcal{I}(\mathbb{V}%
)$ and considers the quotient $\mathcal{R}/\mathcal{I}(\mathbb{V})$ which is
isomorphic to the ring generated by the transfer matrices, namely we saw in
the model discussed here that the eigenvalues of the affine plactic Schur
polynomials for fixed particle number or level $k$ could be identified with
elements in the quotient ring; see Corollary \ref{spectral curve}.

In general it is not true that the Bethe ansatz equations can be solved,
i.e. the variety $\mathbb{V}$ is not known explicitly. Nevertheless one
might still be able to determine the nullstellen ideal $\mathcal{I}(\mathbb{V%
})$ and perform computations in $\mathcal{R}/\mathcal{I}(\mathbb{V})$,
similar to the computations of the fusion coefficients performed in Section
4, where we only used the abstract from of the polynomial equations (\ref%
{bethe ideal}) but not the explicit solutions (\ref{broots}).

From these observations a natural classification question arises: can the
commutative algebras arising from integrable vertex models associated with
the quantum groups $U_{q}\mathfrak{\hat{g}}$ be identified and do their
structure constants have a similar representation theoretic interpretation?
We hope to address this question in future work.\medskip

\textbf{Acknowledgments}. The research of the author is carried out under a
University Research Fellowship of the Royal Society. Part of the results
summarised where presented during the workshop \emph{Infinite Analysis 10}
\textquotedblleft \emph{Developments in Quantum Integrable Systems}%
\textquotedblright\ held at the Research Institute for Mathematical
Sciences, Kyoto, Japan during June 14-16, 2010. The author wishes to thank
the organisers, Professors Michio Jimbo, Atsuo Kuniba, Tetsuji Miwa, Tomoki
Nakanishi, Masato Okado, Yoshihiro Takeyama for their kind invitation and
hospitality.

%%%%%%%%%%%%%%%%%%%%%%%%%%%%%%%%%%%%%%%%%%%%%%%%%%

\end{document}